\documentclass[3p]{elsarticle} %3p,draft
 \myfooter[C]{}
\renewcommand{\thispagestyle}[1]{}

\date{}
\def\texpsfig#1#2#3{\vbox{\kern #3\hbox{\includegraphics{#1}\kern #2}}\typeout{(#1)}}

\usepackage{url,hyperref}
\hypersetup{
    colorlinks=true,
    linkcolor=blue,
    filecolor=magenta,
    urlcolor=cyan,
    pdfpagemode=FullScreen,
    }
\usepackage{bbm}
\usepackage{eurosym}                           % Symbol for Euro
\usepackage[latin1]{inputenc}                  % Input encoding for Western Europe
\usepackage{url}                               % Display URLs in typewriter font
\usepackage{longtable}                         % Tables running of several pages
\usepackage{array}                             % more table options
\usepackage{graphicx,color}
\usepackage{amsthm}
\usepackage{amsbsy}

\usepackage{amssymb}
\usepackage{amsmath}

\usepackage{enumerate}
\usepackage{float}
\usepackage{mathtools}
\usepackage{cleveref}

\usepackage{graphicx,color}
\usepackage{epsfig}
\usepackage[ruled,vlined, linesnumbered]{algorithm2e}
\usepackage{multirow,bigdelim}
\usepackage[table]{xcolor}% http://ctan.org/pkg/xcolor
\usepackage{natbib}
\definecolor{myblue}{RGB}{136, 135, 214} % Define a custom blue color
\theoremstyle{plain}
\newtheorem{thm}{Theorem}[section]

\newtheorem{dfn}[thm]{Definition}
\newtheorem{assumption}{Assumption}

\newtheorem{rem}{Remark}
\theoremstyle{remark}

\theoremstyle{plain}

\newtheorem{lem}[thm]{Lemma}

\theoremstyle{definition}

\def\R{\mathbb{ R}}             % Real number
\def\Q{\mathbb{ Q}}             % Measure Q

\def\N{\mathbb{ N}}

\def\F{\mathcal{F}}  % Filtration
\def\M{\mathcal{M}}  % Set of all sample matrices
\renewcommand{\d}{{\rm d}}      % straight "d" in in integration and ODEs, \int_a^b f(x)\d x
\def\dW{{\rm d}W}               % dW in SDEs- Brownian noise.
\def\dt{{\rm d}t}

\def\1{{\mathbbm{1}}}            % Indicator function

\DeclarePairedDelimiter{\norm}{\lVert}{\rVert} 

\DeclareMathOperator*{\argmin}{arg\,min}
\theoremstyle{plain}% default
\SetKwInput{KwInput}{Input}                % Set the Input
\SetKwInput{KwOutput}{Output} 

\newtheorem{theorem}{Theorem}[section]

\newtheorem{example}{Example}[section]
\usepackage[margin=1cm]{caption}
\usepackage{subcaption}
\numberwithin{equation}{section}	     %Equation numbering per section

\title{ \raggedright  Basket Options with Volatility Skew: Calibrating a Local Volatility Model by Sample Rearrangement}

\begin{document}
\author[1,3]{NICOLA F. ZAUGG\corref{cor1}}
\ead{n.f.zaugg@uu.nl}
\author[1,2]{LECH A. GRZELAK}
\ead{L.A.Grzelak@uu.nl}
\cortext[cor1]{Corresponding author at Mathematical Institute, Utrecht University, Utrecht, the Netherlands.}
\address[1]{Mathematical Institute, Utrecht University, Utrecht, the Netherlands}
\address[2]{Financial Engineering, Rabobank, Utrecht, the Netherlands}
\address[3]{Capital Markets Technology, swissQuant Group AG, {Z\"u}rich, Switzerland}
%{\raggedright ver 1. \today}
\begin{abstract}
The pricing of derivatives tied to baskets of assets demands a sophisticated framework that aligns with the available market information to capture the intricate non-linear dependency structure among the assets. We describe the dynamics of the multivariate process of constituents with a copula model and propose an efficient method to extract the dependency structure from the market. The proposed method generates coherent sets of samples of the constituents process through systematic sampling rearrangement. These samples are then utilized to calibrate a local volatility model (LVM) of the basket process, which is used to price basket derivatives. We show that the method is capable of efficiently pricing basket options based on a large number of basket constituents, accomplishing the calibration process within a matter of seconds, and achieving near-perfect calibration to the index options of the market.
\noindent
\end{abstract}

\maketitle

{\let\thefootnote\relax\footnotetext{The views expressed in this paper are the personal views of the authors and do not necessarily reflect the views or policies of their current or past employers. The authors have no competing interests.}}

%%%%%%%%%%%%%%%%%%%%%%%%%%%%%%%%%%%%%%%%%%%%%%%%%%%%%%%%%%%%
\section{Introduction}
Basket derivatives are a common class of exotic financial instruments. These derivatives depend on the performance of a linear combination of multiple underlying assets, referred to as the constituents of the basket. Constituents typically include equities, currencies, or commodities, although theoretically, they can be any financial asset. Basket products take various forms, ranging from vanilla options to intricate structured products that involve more advanced payoffs and additional interest rate components. They enable trading participants to gain exposure to multiple assets with a single trade, thereby carrying a significant amount of correlation risk.

 Basket derivatives are exotic derivatives and are commonly traded over-the-counter (OTC), meaning they are usually illiquid instruments. Some exceptions are derivatives on stock indices,  commodity indices, or interest rate spreads (hereafter denoted as ``index derivatives"), which are also classified as basket products. Vanilla options on such baskets often exhibit high liquidity and are actively traded on exchanges. For instance, the equity index options volume exceeded the single stock option volume by more than 400\% in 2022 \cite{euronext_data}. Considering a substantial portion of basket products are traded over-the-counter, the development of an efficient and consistent risk-neutral pricing framework becomes crucial for traders trying to offer a fair price to the market and hedge the risks accordingly \cite{qu2010pricing}. Financial modelers encounter two primary challenges in implementing such frameworks. The first challenge involves addressing the inherent high dimensionality of the problem. Because each constituent contributes to the overall performance of the derivative, numerical approximations rapidly become computationally expensive as their calibration and evaluation time increases substantially. An effective pricing framework frequently incorporates efficient approximations and dimensionality reductions to enhance performance as needed.

Secondly, a pricing framework has to be calibrated to the available market data in a consistent way. Once calibrated, the model should replicate the market price of all liquid financial instruments in the scope of the model. It requires enough flexibility to price constituents' implied volatility skew as well as the index volatility skew, for which liquid market data is available.

Suppose that an index $I(t)$ consists of $N$ constituents $S_1(t), S_2(t), \dots S_N(t)$. The price process of the index at time $t$ is then given by
\begin{equation}
\label{eq:index}
    I(t) = \sum_{n=1}^N w_n S_n(t),
\end{equation}
where $w_n$ are the weights of the index. Since the prices of the constituents can be scaled individually, let us, without loss of generality, assume that $w_n = 1, \; \forall n \leq N$, such that the index is the sum of all constituents. Suppose that $N_I = \N_{\leq N}$ is the set of all integers below or equal to $N$. A basket $B(t)$ is formed from any subset $N_B \subset N_I$ of constituents and is also equal to the sum of its (possibly weighted) constituents
\begin{equation}
\label{eq:basket}
    B(t) =  \sum_{n\in N_B}  S_n(t).
\end{equation}
In a risk-neutral setting, the constituents $S_n(t)$ are random processes on a probability space $(\Omega, \F(t), \Q)$, where the dynamics of $S_n(t)$ are determined by the market prices of vanilla options on the constituents.

To price a derivative whose payoff depends on \Cref{eq:basket}, we require information about the distribution of the process $B(t)$. Since the marginal distribution functions of $S_n(t)$ are determined by the individual constituents' options, the difficulty in modeling the distribution of $B(t)$ lies solely in the joint distribution of the constituents. While the essence of such a pricing framework is \emph{not} to price derivatives on index $I(t)$, the availability of liquid index options provides certain market-implied information about the joint distribution between the individual constituents. Extracting this information and embedding it in the model allows us to price basket options on any subset of constituents consistently.

A practical approach to modeling basket distributions is the so-called \emph{correlated local volatility model} \cite{xu2010basket}. The model assumes correlated local volatility diffusion processes for the individual constituents, where the correlations are based on historical data or calibrated to index ATM (At-The-Money) volatilities. The basket process is then simulated using Monte-Carlo methods or derived from some analytical approximations (See, for instance, \cite{Pirjol_2023}).

It is widely acknowledged that such an approach suffers from a key pitfall \cite{Branger_2004,Grzelak_Jablecki_Gatarek_2023,langnau2010dynamic}. As mentioned earlier, to properly capture the volatility skew of the index, the probability distribution of $I(t)$ must align with the option-implied risk-neutral probability distribution. However, the probability distribution assumed by the model is usually inconsistent with its implied distribution by the market. This inconsistency means that the index option smile/skew cannot be faithfully replicated \cite{Grzelak_Jablecki_Gatarek_2023, langnau2010dynamic}. The primary issue lies in the fact that the correlated Brownian motions, which model the linear dependency between the assets, are insufficient to explain the complex dependency structure among the constituents. Consequently, this leads to inconsistencies in the smile/skew, as the calibration procedure can only match ATM volatilities and the model will lead to large mispricing in volatile market conditions \cite{bedendo2010pricing}.  A detailed discussion of this issue will be provided in \Cref{sec:ModelDescription}.

The most notable class of models aimed at resolving the issue of the index skew is the class of so-called local correlation models, extensively discussed in the literature \cite{langnau2010dynamic, Reghai_2010, Guyon_2016}. In these models, the correlation between constituents at time $t$ depends on the state $S_n(t)$ or $I(t)$. While these models effectively address the problem of non-constant correlation, they face significant drawbacks due to their high computational complexity. Specifically, each constituent of the index needs to be simulated using the complex non-constant correlation in a numerical scheme, which leads to inefficiencies. Efforts to resolve this have been made in \cite{Koster_2019}, which promise lower computational costs, but in turn come at the cost of imperfect calibration to the index skew.

Parallel to local correlation models are stochastic correlation models, grounded in correlation as a Jacobi process and further enriched with jumps \cite{Zetocha_2015}. While these models provide insights into the precise nature of the correlation skew among assets, their practical usability is limited due to lower flexibility with similar computational complexity, compared to the local correlation model.

While local and stochastic correlation models can reprice the index options by enhancing the flexibility of the model with extra parameters, at the cost of computational complexity, they only partially capture the underlying issue: correlation is not sufficient to model the joint distribution. Aiming to address this issue more fundamentally, copula models are proposed as an alternative to correlation models (e.g., \cite{salmon2006pricing,van2005bivariate, Lucic_2012}). Copula models leverage a mathematical tool to describe complex dependencies among random variables, capturing the intricate nature of joint distributions. While these models are mathematically intriguing, explicit copula models are limited to bivariate copulae or lack dynamic form and provide no straightforward method for calibration other than a brute-force search for parameters \cite{bedendo2010pricing}.

A feature shared by all mentioned approaches so far is that they propose a multivariate model, meaning that each constituent asset price is modeled explicitly by a stochastic differential equation (SDE). For pricing derivatives on the process $B(t)$, however, it would suffice to model the dynamics of $B(t)$ with a one-dimensional SDE, for instance, by deriving a local volatility model (LVM)
\begin{equation}
\label{eq:modelvol}
    \frac{\d B(t)}{B(t)} = r\dt + \sigma_{LV}(t,B(t)) \dW(t),\;\; t \geq 0,
\end{equation}
where the local volatility function $\sigma_{LV}(t,B(t))$ is given by Dupires formula\footnote{Alternatively, an equivalent formulation using implied volatility can be used. See \cite[Section 4.3.1]{Oosterlee_Grzelak_2020}} \cite{Oosterlee_Grzelak_2020},
\begin{equation}
\label{eq:lvariance}
 \sigma_{LV}^2(t,k) = \frac{\frac{\partial}{\partial t}\left[e^{-rt} \int_k^\infty (y-k) f_{B(t)}(y)\d y\right] + rk \left(\int_{-\infty}^k f_{B(t)}(y)\d y -1\right)}{\frac{1}{2}k^2 f_{B(t)}(k)},
\end{equation}
with $f_{B(t)}$ the probability density function (PDF) of the basket $B(t)$ at time $t$ and $r$ the prevailing interest rate. While the calibration of this model still requires a suitable multidimensional model to derive $f_{B(t)}$ from the component processes, this model of $B(t)$ has an advantage: The estimation of $f_{B(t)}$ at a fixed time $t$ can be done statically, meaning that $f_{B(t)}$ can be determined \emph{independently} for a time discretization $t \in \{t_1,t_2,\dots,t_E\}$ with $E$ time steps. The probability densities for each $t$ can then be interpolated and combined to determine the local volatility function in continuous time. Such a static calibration method was proposed by \citet{Grzelak_Jablecki_Gatarek_2023} where the authors aim to estimate $f_{B(t)}$ by constructing a set of samples of $B(t)$ for all maturity times $t$ available in the market. The sample sets are initialized on a constituent level with an appropriate sampling scheme to enforce a dependency between the constituents and then combined to create a sample of the basket $B(t)$. The probability density function $f_{B(t)}$ is then estimated from the empirical distribution of the samples. The authors thus show that the problem of calibrating the local volatility function of \Cref{eq:lvariance} is reduced to constructing (static) sample sets of $B(t)$ using an appropriate sampling scheme. 

A particular sampling scheme that is suitable for this type of problem are rearrangement algorithms. Rearrangement algorithms are used to construct samples of multivariate random variables when the marginal distributions are given, but only limited information about the joint probability distribution is available. These algorithms have found success in mathematical finance before, most notably in the context of risk management \cite{embrechts2013model}. \citet{bernard2021model} showed that the algorithms are applicable to the index/basket option setting as the constituent PDF $f_{S_n(t)}$ are given by the constituents option market. The constituent samples are first initialized marginally and are then rearranged to ensure that the resulting empirical distribution of the index samples (obtained by summing up the rearranged samples) is ``admissible", meaning that it is close to the probability distribution obtained from the market. The empirical distribution of $B(t)$ is then available by summing up the appropriate constituent samples. Rearrangement algorithms implicitly extract a dependency structure from the available market information by finding an optimal arrangement, such that the samples reflect the available information provided by the cumulative distribution function (CDF) $F_{I(t)}$ of the index. 

In this paper, we will describe the theoretical and practical foundation of utilizing rearrangement algorithms in the context of basket derivatives. Using a copula model we first describe a risk-neutral pricing framework where each constituent process is marginally given by a local volatility model. We then show how rearrangement algorithm are used to extract a market-implied dependency structure to calibrate the copula of the processes. Building on the theoretical foundation, we introduce our own algorithm which we call the Iterative-Sort-Mix algorithm (ISM). This algorithm generates constituent samples for an index with relatively low computation requirements. In the final step, the samples of the algorithm are used to calibrate the local volatility model of the basket $B(t)$ given by (\ref{eq:modelvol}), which can then be used to price vanilla and exotic payoffs on the basket $B(t)$. We assess the effectiveness of the algorithm using real market data for index options on the Dow Jones Industrial Average (DJIA) with 30 underlying constituents. The algorithm successfully prices the volatility skew for all available time to maturities, establishing its suitability as a pricing model for basket derivatives. 

The paper is structured as follows: in \Cref{sec:ModelDescription}, we first describe the general multivariate copula model and explain the important aspects of a suitable dependency structure, particularly highlighting the inadequacy of correlated LVMs. Furthermore, we discuss the problem of underdetermination when calibrating a copula model to the market. In \Cref{sec:Implied}, we provide the mathematical foundation for a rearrangement algorithm to generate the samples required to estimate $f_{B(t)}$, and show how to apply the theory from  \Cref{sec:ModelDescription} to rearrangement algorithms. In \Cref{sec:algorithm}, we introduce the ``Iterative Sort-Mix" algorithm, an implementation of a rearrangement algorithm, which can be used to price basket options. Finally, we apply the algorithm to actual data in \Cref{sec:marketData} and we draw our conclusions in \Cref{sec:concl}.

\section{Copula Models: Joint Distribution of Constituents}
\label{sec:ModelDescription}
\subsection{Model Description}
\label{impliedCorr}
In this paper, we study the joint dynamics of $N$ random processes (called constituent processes) 
\begin{equation}
\label{eq:model}
    \left(S_1(t),S_2(t),\dots,S_N(t)\right),
\end{equation} on a probability space $(\Omega, \F(t), \Q)$, where $\Q$ is the risk neutral measure. We will model the individual marginal processes $S_n(t)$ with a local volatility model given by the local volatility function $\sigma_{n, LV}(t,S_n(t))$ and the constant interest rate $r$ as
\begin{equation}
\label{eq:lvm}
    \frac{ \d S_n(t)}{S_n(t)} = r  \d t + \sigma_{n, LV}(t,S_n(t)) \d W_n(t).
\end{equation}
The Brownian motions $W_n(t)$ are dependent on each other and their dependency structure will be defined later. Given a subset $N_B$ of the $N$ constituents,  the price process of the basket is obtained by summing the individual constituents of the basket.
\begin{equation}
    B(t) =  \sum_{n\in N_B}  S_n(t).
\end{equation}
The process $B(t)$ can also be modeled by a ``direct" stochastic differential equation, given by a local volatility model
\begin{equation}
\label{eq:directModelvol}
    \frac{\d B(t)}{B(t)} = r\dt + \sigma_{LV}(t,B(t)) \dW(t),\;\; t \geq 0.
\end{equation}
As we hinted in the introduction, the key to consistent basket option pricing lies in modeling the dependency structure among constituent assets. While the dynamics of \Cref{eq:lvm} determine the cumulative distribution function of constituent assets, the sum's CDF is dictated by the joint distribution of the constituents and, consequently, by the dependency structure between the individual assets. Suppose that $F_{S_n(t)}$ is the probability distribution function of $S_n(t)$. To describe the joint distribution of the constituents, we will use a (time-dependent) copula, which is given by the function $C \colon [0,\infty) \times [0,1]^N \to [0,1]$, such that
\begin{equation}
    C(t,u_1,u_2,\dots,u_N) = \Q(F_{S_1(t)}(S_1(t)) < u_1, F_{S_2(t)}(S_2(t)) < u_2,\dots,F_{S_N(t)}(S_N(t)) < u_N).
\end{equation}
The copula function is assumed to be sufficiently smooth in the space variables to admit a density copula for all times $t \in [0,\infty)$ \begin{equation} 
c(t,u_1,u_2,\dots,u_N) = \frac{\partial^N  C(t,u_1,u_2,\dots,u_N) }{\partial u_1\partial u_2 \dots \partial u_N }.\end{equation}
We will characterize the set of all copulae of the joint random vector $(S_1(t),S_2(t), \dots, S_N(t))$ as $\Theta$ , such that $\theta \in \Theta$ determines a unique copula function $C_\theta$.

The challenge in pricing basket options lies in selecting a suitable copula function $C$ that is coherent with the market data. Suppose that the market has a certain view on the dependence of an index's constituents. This dependence is reflected in the market-implied risk-neutral probability distribution of $I(t)$ at time $t \in [0, \infty)$. The availability of liquid index option prices allows us to extract this probability distribution, which we call $F_{I^\text{Mkt}(t)}(x)$. The copula $C$ must thus be chosen in accordance with this market expectation \footnote{This is equivalent to saying that the model replicates the implied volatility skew of the market, meaning that vanilla options prices from the model will be equal to the market prices.}, i.e. 
\begin{equation}
\label{eq:rnd}
    F_{I(t)}(x) = \Q(I(t) < x) = \Q\left(\sum_{n=1}^N S_n(t) < x\right) = F_{I^\text{Mkt}(t)}(x), \hspace{0.5cm} \forall x \in \R,\; t \in [0,\infty).
\end{equation}
We call a copula $C$ which agrees with the market-implied distribution an \emph{admissible copula}.
\begin{dfn}[Admissible Copula]
\label{def:implied}
Suppose that $S_1(t),S_2(t), \dots S_N(t)$ have a joint distribution with copula $C$ such that \Cref{eq:rnd} holds.
Then, we call the copula function $C$ an admissible copula. We denote the set of all admissible copulae as $\Theta_A \subset \Theta$.
\end{dfn}
Note that the existence of an admissible copula depends on an assumption of efficiency in the market. The distribution $F_{I^\text{Mkt}(t)}$ is obtained from index option quotes, while we calibrate the dynamics of the individual constituents to single constituents' option quotes. We will assume that the two sets of option quotes are consistent with each other, implying that there is no arbitrage in the market.  
\begin{assumption}[Market Efficiency]
\label{assump:1}
We will assume that the index options priced in the market are in accordance with the vanilla options of its constituents, and thus, that an admissible copula can be found.
\end{assumption}
So far the introduced model is a time-continuous model for the constituent and index processes, meaning that the copula $C(t, \cdot)$ is also defined on a time continuum $[0, \infty)$. The data from the market for $S_n(t)$ and $I(t)$ to calibrate the model is generally only available for a discrete set of market maturities $T_E = \{t_1,t_2,\dots,t_E\}$. This means that we can only calibrate an admissible copula to this discrete set of time points and an interpolation/extrapolation in time after calibrating the model. The chosen time interpolation then determines the exact shape of the copula on the time continuum. Such a time interpolation can be done on the implied volatilities and is standard practice (see for instance \cite{Oosterlee_Grzelak_2020,andreasen2011volatility}). This approach, however, has an impact on our model. The copula, which is defined on a time-continuous spectrum, is thus fully determined by its value at the discrete market maturities $T_E$. This means that the following additional assumption is made:
 \begin{assumption}[Time interpolation] If $C_1$ and $C_2$ are two copulae, we have
\[C_1(t, \cdot) = C_2(t,\cdot), \; \forall t \in [0,\infty) \iff C_1(t,\cdot) = C_2(t,\cdot), \;  \forall t \in T_E.\]
 \end{assumption}
\subsection{Intricate Dependency Structures}
Before diving further into the copula model, let us first investigate why choosing a simple dependency structure, such as in the correlated local volatility model (cLVM) is not sufficient to price index and basket options. The dependence of the constituents in a cLVM is directly modeled by correlating the Brownian motions in \Cref{eq:lvm}. This means that for any distinct pair $n, m \in N_I$, we have
\begin{equation}
\label{eq:correlation}
\d W_n(t) \d W_m(t) = \rho_{n,m} \d t, \end{equation}
where $\rho_{n,m} \in [-1,1]$ is the correlation\footnote{Note that the correlation matrix has to be positive definite.} between the Brownian motion. The price processes are thus driven by correlated Brownian motions with correlation $\rho_{n,m}$, and we can price index options with a Monte-Carlo pricer by simulating $S_n(t)$ for each $n$. 

We observe now, when calibrating the correlations to the index option prices, that the fitted correlation coefficients of the constituents $\rho_{n,m}$ are not equal when calibration to options of different strikes \cite{Branger_2004}. It appears as if the market is assuming different correlation coefficients of the underlying assets when considering different strikes. This observation is referred to as ``implied correlation", which, similar to implied volatility, describes the strike-dependent correlation coefficient observed in the market. In practice, the existence of implied correlation means that if a model is calibrated to ATM prices of the index option, non-ATM options will be mispriced by the model. 

Just as with the existence of an implied volatility skew, the existence of implied correlation is not due to irrationality in the market, but rather due to the model's inability to capture all the complex effects present in the market. It turns out that the dependence structure of \Cref{eq:correlation}, which is solely determined by the correlation coefficient, induces a \emph{linear} relationship between the assets. This linear relationship means that the dependency of the assets does not depend on the absolute values. This is insufficient to model dependencies for constituents. On one hand, analysis of historical correlations shows that correlation is stronger for negative returns than for positive returns \cite[Figure 1]{Zetocha_2015}. Such a structure is also expected by the market, as the market-implied expectation of the correlation often exhibits a skew or smile \cite[Figure 2]{Zetocha_2015}. 

We visualize the impact of the difference in dependency structures in \Cref{fig:corr-copula}. The dependence between two assets at a fixed time $t$ can be visualized with a scatter plot. While the Pearson's correlation coefficient is $0.5$ in both cases, the left shows a non-linear dependency structure. The skewed distribution exhibits a strong left-tail correlation, meaning that the correlation is stronger for lower values.
\begin{figure}[H]
    \centering
    \includegraphics[width=\textwidth]{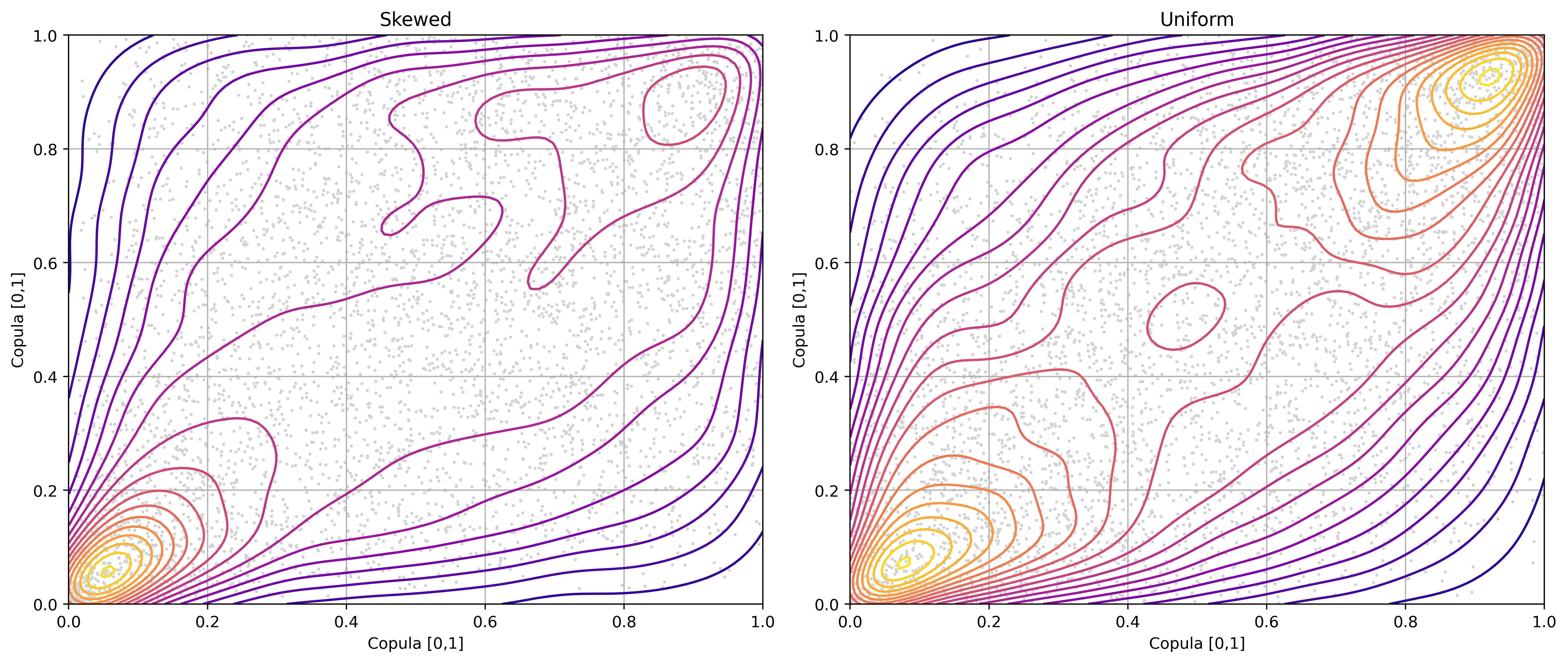}
    \caption{Example of a bivariate copula, visualized. Left: Stronger correlation for low values. Right: Uniform correlation. In both cases, the Pearson correlation is $0.5$.}
    \label{fig:corr-copula}
\end{figure}
It is not obvious at first sight why the exact structure of the correlation matters for pricing options. The impact of the structure is best observed by examining the generated PDF of the index $I(t)$ under uniform correlation vs non-uniform correlation. \Cref{fig:unifromvsNonun} displays the distribution of the sum of the random variables in \Cref{fig:corr-copula} and the impact on implied volatility skews given these correlation structures. The non-uniform correlation structure with a stronger lower-tail correlation (\Cref{fig:corr-copula}, left) will exhibit a broader lower tail, resulting in a skew in the implied volatility.
\begin{figure}[H]
    \centering
    \includegraphics[width=\textwidth]{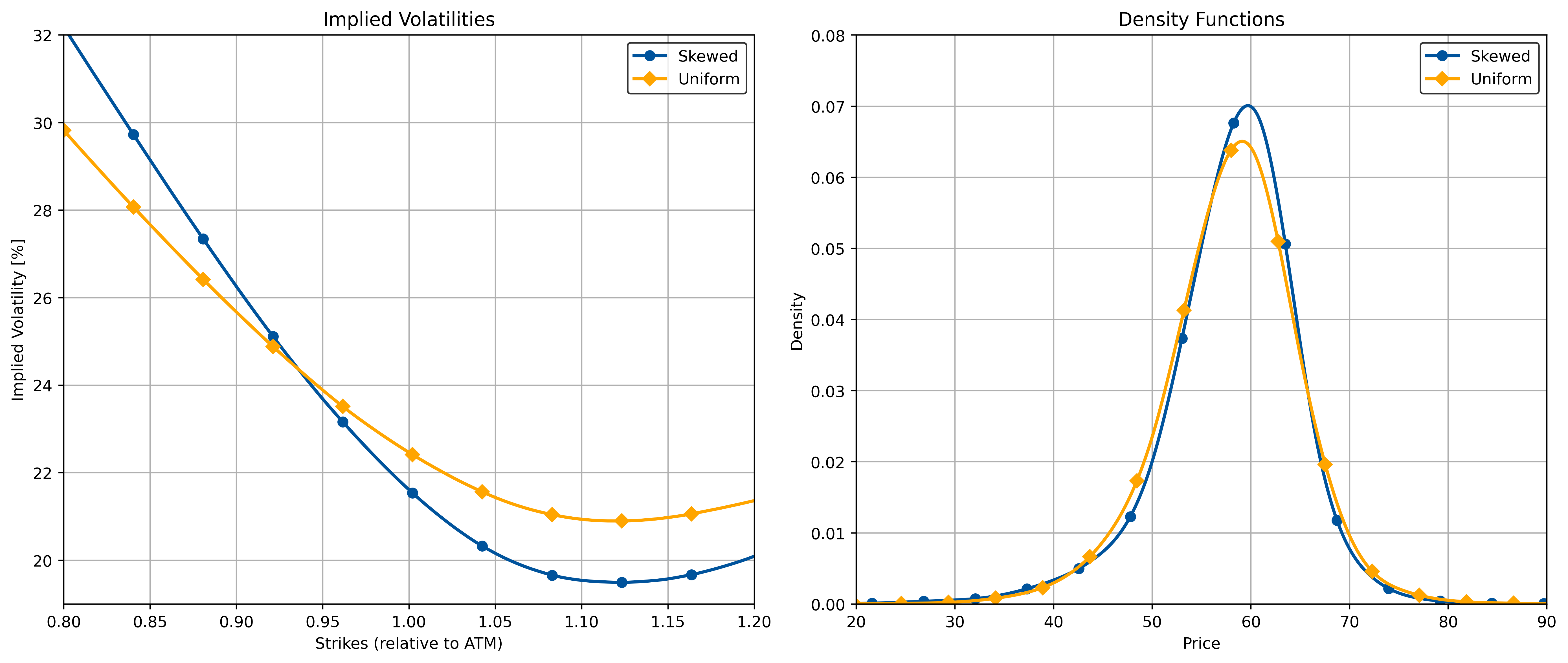}
    \caption{Effect of difference in copula on basket pricing: We plot the PDFs and the implied volatility skew of the sum of two random variables as in \Cref{fig:corr-copula}. The copula with a strong left tail exhibits a larger implied volatility skew. }
    \label{fig:unifromvsNonun}
\end{figure}
Given that the market expects a non-linear correlation structure, it is therefore crucial to have a model that is flexible enough to account for this expectation to reprice the index options. Copula models allow this flexibility since copulae can describe any type of dependency between $N$ random variables. 

The obvious way to utilize a copula model in practice is to specify an admissible copula $C$, which then defines the distribution of $I(t)$ and $B(t)$ and any time $t$. Such an approach was described for instance by \citet{Lucic_2012} using product copulae. The main issue with such an approach is the calibration of the copula. It turns out it is increasingly difficult to specify a copula for a high number of constituents, and often, a brute-force search algorithm is required to find the right parameters.

For this reason, rearrangement algorithms are used as an alternative to describe the dependency. Instead of defining an explicit copula $C$ which is suitable to model the correlation, rearrangement algorithms implicitly extract the dependency structure and generate samples for the random vector (\ref{eq:model}) under an appropriate admissible copula. These samples are sufficient to calibrate \Cref{eq:lvm}, since the cumulative distribution function $F_{B(t)}$ can be inferred directly from these samples. The exact workings of rearrangement algorithms will be described in \Cref{sec:Implied}.
\subsection{Uniqueness of Admissible Copula}
\label{sec:unique}
We defined an admissible copula $C$ to be any copula for which \Cref{eq:rnd} is given. The condition ensures that the model captures the market expectation of the copula $C$ given the information available from the index option market. Since the condition (\ref{eq:rnd}) can be fulfilled for more than one copula $C$, we should consider the implications of multiple distinct admissible copulae. At first sight, the uniqueness might not seem important, as all admissible copulae lead to the same index probability distribution, and thus also to the same index option prices. However, while the index option prices are the same for a model with any admissible copula, the resulting basket probability distributions given a basket subset $N_B$ are not necessarily consistent. This means that two admissible copulae can imply different basket probability distributions for the subset of constituents. The source of this issue is in the fact that $F_{I^\text{Mkt}(t)}$ only provides partial information about the dependency of the constituent random variables. We refer to this as the problem of \emph{underdetermination} of copula $C$. To illustrate the problem, let us examine a toy model of three uniformly distributed assets.
\begin{example}
\label{ex:3assets}
Let $t \in [0,\infty)$ be fixed and let $S_1(t),S_2(t),S_3(t)$ be marginally uniformly distributed on $[0,1]$. Suppose in a first case that $S_1(t)=S_2(t)$ and $S_3(t)$ is independent of them. In this case, the copula $C_1(t)$ is given by
\[C_1(t,u_1,u_2,u_3) = \Q(S_1(t) \leq \max(u_1,u_2)) \Q(S_3(t)\leq u_3) =   \max(u_1,u_2) u_3.\]
Suppose now that  $S_1(t)=S_3(t)$ and $S_2(t)$ is independent of them. In this case
\[C_2(t,u_1,u_2,u_3) = \Q(S_1(t)\leq \max(u_1,u_3)) \Q(S_2(t)\leq u_2) =   \max(u_1,u_3) u_2.\]
The copula $C_1$ and $C_2$ are not equal to each other, but it is clear that in either case, the distribution of $S_1(t)+S_2(t)+S_3(t)$ is $2U_1 + U_2$, where $U_1,U_2$ are independent uniformly distributed random variables. The copula $C_1$ is thus admissible if and only if $C_2$ is admissible. We run into an issue when considering a basket $N_B = \{1,2\}$, as the distribution of $S_1(t)+S_2(t)$ is not the same under $C_1$ and $C_2$. In the first case, we have
\[B(t) = 2U_1,\]
while in the second case,
\[B(t) = U_1 + U_2.\]
Pricing a derivative on $B(t)$ will yield different results in the cases.
\end{example}
The example clearly shows that an admissible copula is not unique, and the distinct copulae lead to different dynamics of $B(t)$. This poses the question of how to determine which copula is the ``superior" one, yielding the ``right" basket option price. Given the lack of additional information on the dependency between the constituents, a reasonable approach is to select a copula with a symmetry property. 
\begin{dfn}[Symmetric Copula]
Let $C \colon [0,\infty) \times [0,1]^N$ be a copula. Suppose that $\pi \colon N_I \to N_I$ is a permutation of the vector indices, and $P_\pi \colon [0,1]^N \to [0,1]^N$ the function which rearranges a vector according to $\pi$, i.e
\[P_\pi(u_1,u_2,\dots,u_N) = \left(u_{\pi(1)},u_{\pi(2)},\dots,u_{\pi(N)}\right).\]
The copula $C$ is called symmetric iff:
\[C(t,u_1,u_2,\dots,u_N)=C(t,P_\pi(u_1,u_2,\dots,u_N)),\]
for all $t \in  [0,\infty)$. 
\end{dfn}
A symmetric copula is a copula where the dependency structure between any two distinct constituents is the same. In other words, we do not favor the dependency between any pair within the constituents of the index, since we do not have any information on such dependency. This is in contrast to the two copulae chosen \Cref{ex:3assets}, where $S_1(t)$ and $S_2(t)$ are strongly dependent in the first case, while $S_3(t)$ is independent of them (and vice-versa in Case 2).  
\begin{assumption}
    We consider an admissible copula $C \in \Theta_A$ as ``valid" if the copula is symmetric. We denote the set of all symmetric copulae as $\Theta_A^S$.
\end{assumption}
Note that most basket option models implicitly only consider symmetric copulae by construction, although such terminology is not explicitly introduced (for instance \cite{Grzelak_Jablecki_Gatarek_2023, Lucic_2012, Zetocha_2015}). The concept of symmetry is also discussed in \citet{bernard2021model} in the context of ``entropy".
\section{Rearrangement Algorithms}
\label{sec:Implied}
In the preceding section, we defined the copula model for a multivariate process for constituents of an index or basket. We demonstrated that the difficulty of such a model is the choice of the copula $C$, such that $(S_1(t), S_2(t), \dots, S_N(t))$ has a joint probability distribution with an admissible and symmetric copula. In the remaining section, we aim to utilize the theoretical model and convert it into a practical pricing framework for any applicable payoff on $B(t)$. In particular, we aim to calibrate the ``direct" local volatility model for $B(t)$ given by \Cref{eq:directModelvol}. 

Rather than directly specifying a copula for the multivariate model, we will extract the dependency structure of the multivariate model by using a rearrangement sampling algorithm. This sampling algorithm will provide samples for the constituents given the marginal distribution $F_{S_n}(t)$ of the constituents, and the index CDF $F_I(t)$. The calibration of the dynamic of $B(t)$ is then achieved by estimating the probability density function $f_B(t)$ from the samples and approximating the probability distribution using the empirical distribution. \Cref{fig:flowchart} shows a complete overview of the process.
\begin{figure}[H]
    \centering
    \includegraphics[width=\linewidth]{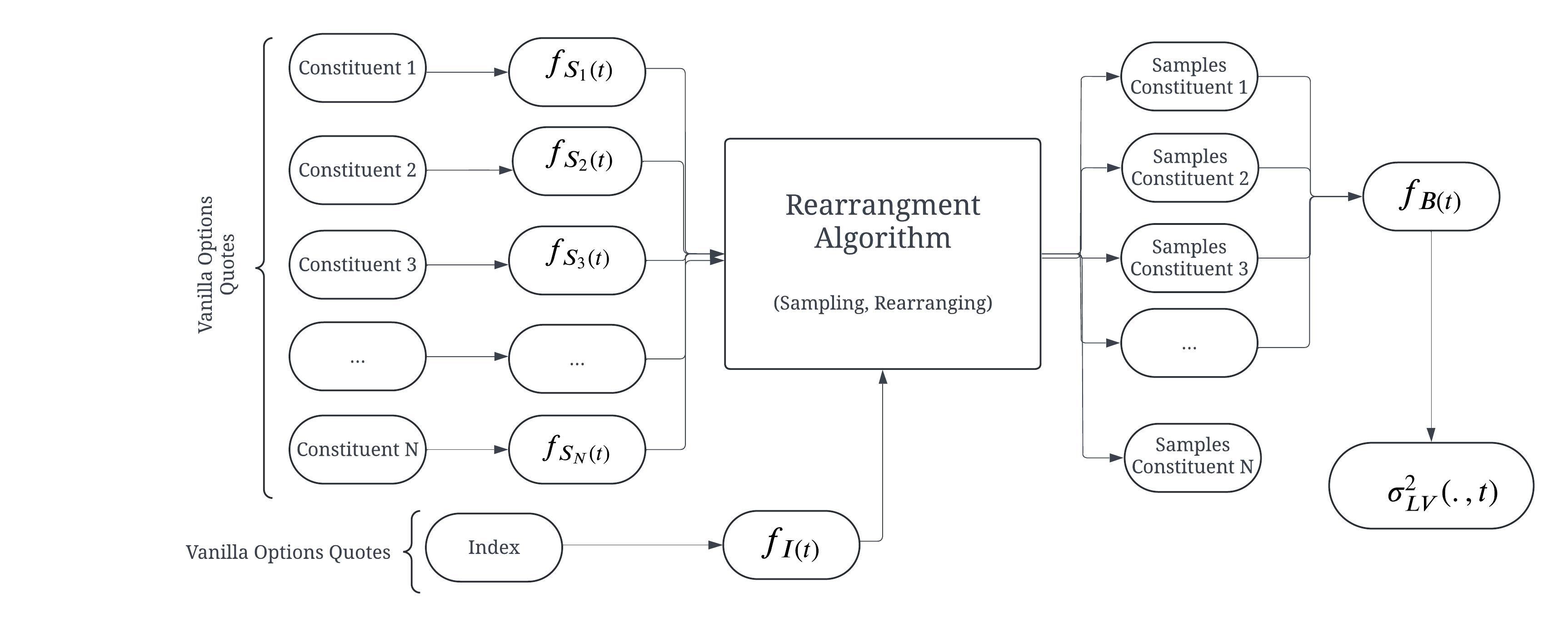}
    \caption{Rearrangement algorithm for basket options: Shows how the rearrangement algorithm is used to calibrate the local volatility model. In the first step, the constituents' option quotes are converted to risk-neutral PDFs. Similarly, we obtain the index PDF from the index option quotes. The rearrangement algorithm then generates samples for the constituents. These samples are used to estimate $f_{B(t)}$, which in turn is used to calibrate the local volatility model of $B(t)$.}
    \label{fig:flowchart}
\end{figure}
The sampling strategies we will consider involve rearrangement algorithms. In this section, we will formally introduce a generic rearrangement algorithm and demonstrate how it can be utilized to generate the necessary samples. Subsequently, in \Cref{sec:algorithm}, we will introduce a specific algorithm designed for calibrating basket models.
\begin{rem}
\label{rem:path-dep}
    The rearrangement algorithm provides samples for $B(t)$ for every $t \in T_E$. We can use these samples to price payoffs on $B(t)$ which only depend on the PDF of $B(t)$ at a single $t \in T_E$, such as vanilla European options. Any path-dependent payoff requires a stochastic model of $B(t)$. The local volatility model of \Cref{eq:directModelvol} is a natural choice given that the local variance function $\sigma^2_{LV}$ can be estimated from the static samples.
\end{rem}
\subsection{General Sample Notation}
We begin by introducing notation for collections of samples for constituents, baskets, and indices. Suppose that, for an asset $S_n(t)$ at some time $t \in T_E$, we draw $M$ samples. We denote the obtained samples as
\begin{equation}
\label{eq:sampleset}
s_n(t) = \left [s^1_n(t),s^2_n(t), \dots s_n^M(t)\right]^T \in \R_+^M.
\end{equation} 
We call the vector $s_n(t)$ a \emph{constituent sample vector} if the asset $S_n(t)$ belongs to an index. Since an index consists of multiple constituents, we require multiple constituent sample vectors to construct samples for an index. These constituent sample vectors are then collected an \emph{index sample matrix}\footnote{We will often refer to it simply as ``sample matrix''. }, denoted as the collection $[s]_M$,
\begin{equation}
    [s]_M= \left[s_n(t) \colon n = 1,2 ,\dots,N \right] = \begin{bmatrix}
s_1^1(t) & s_2^1(t) & \cdots & s_N^1(t) \\
s_1^2(t) & s_2^2(t) & \cdots & s_N^2(t) \\
\vdots & \vdots & \ddots & \vdots \\
s_1^M(t) & s_2^M(t) & \cdots & s_N^M(t)
\end{bmatrix}
\in \R_+^{M\times N} .
\end{equation}
The index sample matrix thus contains all sample vectors from all the constituents. Since the index $I(t) = \sum_{n \leq N}S_n(t)$ is a sum of all the assets in the index, we can define the samples of the index as the sum of the constituent samples. We define $i(t)$ as the sum of the constituent sample vectors:
\begin{equation}
     i(t)  = \sum_{n \leq N }s_n(t) \in \R_+^M.
\end{equation} 
In the same fashion, we find the basket samples for an arbitrary basket $N_B$ by summing up the corresponding constituent samples:
 \begin{equation}
    b(t) =  \sum_{n \in N_B }s_n(t) \in \R_+^M .
\end{equation} 
Lastly, we introduce the empirical distributions given a sample matrix $[s]_M$. Empirical distribution functions count the number of samples observed in a certain region given a sample matrix. For individual constituents, we define the marginal empirical distribution as
\begin{equation}
    \hat{F}_{[s]_M;n}(x) = \frac{1}{M}|\{ m \leq M \colon s_n^m(t) \leq x\}|, \hspace{0.3cm} x \in \R_+,
\end{equation}
where $|A|$ denotes the cardinality of a set $A$.
Furthermore, the empirical distribution for the basket and index is defined as
\begin{equation}
    \hat{F}_{[s]_M;B}(x) = \frac{1}{M}|\{ m \leq M \colon b^m(t) \leq x\}|, \hspace{0.3cm} x \in \R_+,
\end{equation}
and respectively
\begin{equation}
    \hat{F}_{[s]_M;I}(x) = \frac{1}{M}|\{ m \leq M \colon i^m(t) \leq x\}|, \hspace{0.3cm} x \in \R_+.
\end{equation}
We can also generalize the concept of a copula into its empirical counterpart. 
The empirical copula counts all the samples that lie below a certain boundary in the $[0,1]^N$ sample space.
\begin{align}
    \hat{C}_{[s]_M}(u)=\frac{1}{M}\left|\{ m \leq M \colon \left(F_{S_1(t)}(s^m_1(t)), F_{S_2(t)}(s^m_2(t)), \dots ,F_{S_N(t)}(s^m_N(t))\right) \in A(u)\}\right|,
\end{align}
with $u \in [0,1]^N$ and $A(u) = [0,u_1] \times [0,u_2] \times \dots \times[0,u_N]$.

Having defined the empirical copula, it is possible to define the symmetric property of a sample matrix. We call a matrix $[s]_M$ a symmetric sample matrix if $\hat{C}_{[s]_M}(u) = \hat{C}_{[s]_M}(P_\pi(u))$ for any permutation $P_\pi$.
\subsection{General Rearrangement Algorithms}
Let the time $t \in T_E$ be fixed. A rearrangement algorithms are used to construct multivariate samples of a random vector, for which:
\begin{enumerate}[i)]
\itemsep0em 
    \item The marginal probability distributions $F_{S_n(t)}$ are given,
    \item The probability distribution $F_{I(t)}$ of the sum $I(t) = \sum_{n=1}^N S_n(t)$ is given.
\end{enumerate}
Let $\Sigma$ denote the collection of the above involved (input) probability distributions. We define a rearrangement algorithm as a map $\mathbf{A}(M; \Sigma, \omega) = [s]_M^*$, which produces a sample matrix $[s]_M^*$ of size $M$ given the sampling outcome $\omega \in \Omega$. The output $[s]_M^*$ of the rearrangement algorithm should have the following properties: Let $\norm{\cdot}_\infty$ be the $L_\infty$-norm and let $\epsilon > 0$ be fixed. Then,
\begin{equation}
    \label{eq:cond1}
    \norm{\hat{F}_{[s]_M^*;n} - F_{S_n(t)}}_\infty < \epsilon,
\end{equation}
for all $n \leq N$. 
\begin{equation}
    \label{eq:cond2}
    \norm{\hat{F}_{[s]_M^*; I} - F_{I(t)}}_\infty < \epsilon.
\end{equation}
In this case, we refer to $[s]_M^*$ as an \emph{$\epsilon$-admissible sample matrix}. 

A general rearrangement algorithm $\mathbf{A}$ works based on the following principle: 
the algorithm initializes a sample matrix $[s]^\text{init}_M=[s]^\text{init}_M(\omega) $ of the $N$ constituents given a time $t\in T_E$. This is accomplished by drawing samples independently according to the marginal distribution $F_{S_n(t)}$ for each $n \leq N$. Such a sampling procedure is relatively cheap since no dependency between the random variables is assumed. Subsequently, the algorithm rearranges the order of the sample vectors $s^m_n(t)$ for each constituent $s_n(t)$ to obtain a new sample matrix $[\bar{s}]_M$. The rearrangement affects the sample vector $i(t)$ and therefore the empirical distribution $\hat{F}_{[\bar{s}]_M; I}$, while leaving the empirical distribution $\hat{F}_{[\bar{s}]_M; n}$ for the constituent $n$ unchanged. The ordering per constituent $s_n(t)$ is changed until we find $[s]_M^*$ such that the empirical distribution converges to $F_{I(t)}$, which means that condition (\ref{eq:cond2}) is satisfied. Since the empirical marginal distributions $\hat{F}_{[s]_M^*; n}$ are unchanged, the condition (\ref{eq:cond1}) is satisfied too. In the next subsection, we will formalize the procedure of rearrangement. 
\subsection{Rearranging the constituent sample vector}
\label{sec:arranging}
Suppose that we obtain a sample matrix $[s]_M^\text{init}$, drawn from the marginal distributions of $S_1(t),S_2(t),\dots ,S_N(t)$ for some time $t\in T_E$, independently of each other. For any $n$, the sample vector is the $M$-dimensional vector
\begin{equation}
    s_n(t) = \left[s_n^1(t),s_n^2(t), \dots,s_n^M(t)\right]^T \in \R_+^M.
\end{equation}
To alter the order of the vector, we will define a permutation on the indices of the elements. A permutation is defined as a bijective map $\bar{\pi}\colon \N_M  \to \N_M $ on the set $\N_M = \{1,2,3,\dots,M\}$. We can apply the permutation to the sample order to obtain a new constituent sample vector
\begin{equation}
    \Bar{s}_n(t) = \left[s_n^{\bar{\pi}(1)}(t),s_n^{\bar{\pi}(2)}(t), \dots,s_n^{\bar{\pi}(M)}(t)\right]^T \in \R_+^M.
\end{equation}
To simplify the notation of the permutation $\bar{\pi}$ when applying to a sample vector $s_n(t)$, we will use notation $\pi \colon \R^m_+ \to \R^m_+ $ to denote
\begin{equation}
    \pi(s_n(t)) =  \left[s_n^{\bar{\pi}(1)}(t),s_n^{\bar{\pi}(2)}(t), \dots,s_n^{\bar{\pi}(M)}(t)\right]^T.
\end{equation}
Since we can apply permutations to every constituent sample vector in a sample matrix, we use $N$ permutations $\pi_1,\pi_2,\dots,\pi_N$ to obtain a new sample matrix 
\begin{equation}
\label{eq:perSet}
    [\bar{s}]_M = [ \pi_n(s_n(t)) : n \leq N].
\end{equation}
From the initial sample matrix $[s]_M^\text{init}$, we can define a finite\footnote{The collection is finite since the number of unique permutations on $\N_M$ is finite} collection of sample matrices by applying all possible permutations on the set. We define
\begin{equation}
    \M  =\M([s]_M^\text{init}, \omega) =  \{  [\bar{s}]_M : \pi_1,\pi_2,\dots,\pi_N \text{ are permutations on }[s]_M^\text{init}(\omega) \},
\end{equation}
with $[\bar{s}]_M$ as in \Cref{eq:perSet}.
\begin{example}
For example, consider the sample matrix $[s]_M$ shown in \Cref{fig:perm}, where the second column on the right-hand side is shifted up by 1. This means that $\bar{\pi}_2(x) = x+1 \mod M$ and $\bar{\pi}_n = id$ for all $n \neq 2$.
\begin{table}[H]
\centering
  \begin{subfigure}{0.43\linewidth}
    \centering
    \begin{tabular}{c|>{\columncolor{myblue!10}}c|c|c|c|c|c|}
      \multicolumn{1}{c}{} & \multicolumn{5}{c}{Constituents}\\
      \cline{2-6}
      \multirow{5}{*}{\rotatebox[origin=c]{90}{\parbox[c]{2cm}{\centering Samples}}}&  &\cellcolor{myblue!10}1&\cellcolor{myblue!10}2&\cellcolor{myblue!10}\dots &\cellcolor{myblue!10}N\\
      \cline{2-6}
      & 1& \(s_{1}^1\) & \(s_{2}^1\) & \(\ldots\) & \(s_{N}^1\) \\
      \cline{2-6}
      &2&\(s_{1}^2\) & \(s_{2}^2\) & \(\ldots\) & \(s_{N}^2\) \\
       \cline{2-6}
     & \vdots &\(\vdots\) & \(\vdots\) & \(\ddots\) & \(\vdots\) \\
       \cline{2-6}
     & M &\(s_{1}^M\) & \(s_{2}^M\) & \(\ldots\) & \(s_{N}^M\) \\
       \cline{2-6}
    \end{tabular}
    \caption{Sample matrix $[s]_M$}
  \end{subfigure}
  \begin{subfigure}{0.43\linewidth}
    \centering
    \begin{tabular}{c|>{\columncolor{myblue!10}}c|c|c|c|c|c|}
      \multicolumn{1}{c}{} & \multicolumn{5}{c}{Constituents}\\
      \cline{2-6}
      \multirow{5}{*}{\rotatebox[origin=c]{90}{\parbox[c]{2cm}{\centering Samples}}}&  &\cellcolor{myblue!10}1&\cellcolor{myblue!10}2&\cellcolor{myblue!10}\dots &\cellcolor{myblue!10}N\\
      \cline{2-6}
      &1&\(s_{1}^1\) & \cellcolor{red!30}\(s_{2}^2\) & \(\ldots\) & \(s_{N}^1\) \\
      \cline{2-6}
      &2&\(s_{1}^2\) &  \cellcolor{red!30}\(s_{2}^3\) & \(\ldots\) & \(s_{N}^2\) \\
      \cline{2-6}
      &\vdots&\(\vdots\) & \(\vdots\) & \(\ddots\) & \(\vdots\) \\
      \cline{2-6}
      &M&\(s_{1}^M\) &  \cellcolor{red!30}\(s_{2}^{1}\) & \(\ldots\) & \(s_{N}^M\) \\
      \cline{2-6}
    \end{tabular}
    \caption{Permutated matrix by shifting column 2}
  \end{subfigure}
  \caption{Example Permutation: In this permutation, only the second column $s_2(t)$ is shifted. Every sample index increases by 1 for this column. The other columns are unchanged.}
  \label{fig:perm}
\end{table}
\end{example}
\subsection{Optimization Formulation}
The aim of a rearrangement algorithms $\mathbf{A}$ is to construct an $\epsilon$-admissible sample. We will now formulate this in an optimization problem over the possible permutation. First, we define the objective function $L$ such that
\begin{equation}
L([s]_M) = \norm{\hat{F}_{[s]_M,I} - F_{I(t)}}_\infty  = \sup_{x \in \R} \left|\hat{F}_{[s]_M,I}(x) - F_{I(t)}(x) \right|.\end{equation}
Consequently, a rearrangement algorithm $\mathbf{A}$ is defined by the minimization of the objective function $L$ over the set $\M([s]_M^{\text{init}})$:
\begin{equation}
\label{eq:algorithm1}
    \mathbf{A}(M; \Sigma, \omega)  = \argmin_{[s]_M \in \M([s]_M^{\text{init}},\omega)} L([s]_M)= [s]_M^*.
\end{equation}
The above formulation aims to find the minimum of $L$ in the permutations. Previously we formulated the conditions of $[s]_M^*$ in terms of the arbitrarily small error $\epsilon$. There is thus a requirement for convergence of the algorithm, which is justified with the following lemma. Since the proof of the theorem is tedious, we leave it to the appendix. 
\begin{lem}
\label{lem:nonempty}
    Let $\epsilon > 0$ be fixed and denote $\M$ as the set of all permutations of an initial sample matrix $[s]_M^{\text{init}}$ of size $M$. There exists an $M$, such that $\M$ almost surely contains an $\epsilon$-admissible sample matrix.
\end{lem}
\begin{proof}
A complete proof is found in \ref{app:proof}.
\end{proof}
The lemma proves that it is possible for the rearrangement algorithm to converge to a solution and that the error $L$ can be brought arbitrarily low by increasing the sample size. There is, however, still an important open question to consider: Are the multiple ways to arrange the samples to obtain $\epsilon$-admissibility? Remember that from \Cref{ex:3assets}, it is known that there is no unique copula of a random vector that can match the distribution function $F_{I(t)}$. It turns out that the algorithm suffers from the same problem of underdetermination as the copula model from \Cref{sec:ModelDescription}. It is possible to obtain two sample matrices $[s]_M$ and $[s]_M'$ with similar empirical index probability distributions, although the empirical copulae diverge drastically. This is a direct result of the convergence of empirical distributions to the cumulative probability distribution \cite{tucker1959generalization}, and the observations from the previous section. The implication of this observation is that the algorithm possibly generates samples for a different random vector $(S'_1(t),S'_2(t), \dots,S'_N(t))$, such that
\begin{equation}
    \sum_{n=1}^N S_n(t) \stackrel{d}{=} \sum_{n=1}^N S'_n(t)
\end{equation}
The issue of underdetermination can be tackled by adding penalization for undesired properties to the sample sets in the optimization. In \Cref{sec:ModelDescription} we define a ``valid" copula to be admissible and symmetric. Suppose we define a measure of sample symmetry $\mathcal{S}([s]_M) >0$\footnote{A symmetry measurement can be based on the empirical copula.}, such that $\mathcal{S}([s]_M) > \mathcal{S}([s]'_M)$ implies that $[s]_M$ closer to being symmetric. Then, we define the optimization of the algorithm as:
\begin{equation}
\label{eq:algorithm3}
    \mathbf{A}(M; \Sigma, \omega) = [s]_M^* = \argmin_{[s]_M \in \M([s]_M^{\text{init}})} L([s]_M) +  \frac{\lambda}{\mathcal{S}([s]_M)},
\end{equation}
where $\lambda > 0$ is a penalization coefficient. The rearrangement algorithm is thus valid if $[s]_M^* = \mathbf{A}(M, \Sigma, \omega)$ is $\epsilon$-admissible, i.e. converges to the market-implied probability distribution, and is symmetric. In this case, $[s]_M^*$ can be considered a sample set of a copula $C \in \Theta_A^S$.

\subsection{Rearrangement Algorithms for Basket Derivatives}
Rearrangement algorithms are well applicable in the context of basket derivatives. The marginal distributions for the rearrangement algorithm are the constituent distributions obtained from the local volatility models of (\ref{eq:lvm}) and the probability distribution from the sum is the index probability distribution. In this case, we can apply a rearrangement algorithm for any $t \in T_E$ to obtain samples of $(S_1(t), S_2(t), \dots, S_N(t))$ and thereby implicitly extracting a copula $C \in \Theta_A$. Since we require samples for all times $t \in T_E$, we create a collection of rearrangement algorithms.
\begin{equation}
    \mathbf{A}(M,t; \Sigma, \omega) = \{\mathbf{A}(M; \Sigma(t), \omega ) : t \in T_E \}
\end{equation}
Since the samples are created individually for each $t$ they are static, meaning that there is no time continuity in the samples. As explained in \Cref{rem:path-dep},  the pricing of path-dependent requires a calibration of a dynamic model for $B(t)$ as explained in \Cref{fig:flowchart}. We can, however, price European call and put options since their payoff depends only on the final time $t$ at expiry.
\begin{theorem}
\label{thm:1}
Let $[s]_M^*$ be an $\epsilon$-admissible sample matrix of constituents for time $t \in T_E$ and $\epsilon>0$. Furthermore, suppose that $F_{I(t)}$ and $\hat{F}_{[s]_M^*; I}$ agree outside an interval $[a,b] \subset \R^+$. Then, 
    \begin{equation}|V_\text{Call}^{MC}([s]_M^*;K,t) - V_\text{Call}^{\text{Market}}(K,t)| < \epsilon (b-a), \end{equation}
where $V_\text{Call}^{MC}([s]_M^*;K,t)$ is the Monte Carlo price of a European call option given the sample matrix $[s]_M^*$.
\end{theorem}
\begin{proof}
For simplicity, we will assume a constant interest rate of $r=0$. In the case of non-zero interest rates, a constant discounting factor has to be included. Since $\hat{F}_{[s]_M^*; I}$ is not differentiable, we use Riemann-Stieltjes integrals to express the pricing integrals. The market price of a call option with strike $K$ expiring at $t$ can be written as
    \begin{equation*}
        V^{\text{Market}}_\text{Call}(t,K) = \int_{-\infty}^\infty (x-K)^+ \d F_{I(t)}(x),
    \end{equation*}
    and the Monte-Carlo price as
        \begin{equation*}
      V_\text{Call}^{MC}([s]_M^*;K,t) = \sum_{j \leq M} \frac{\left(i^j(t) - K\right)^+}{M} = \int_{-\infty}^\infty (x-K)^+ \d \hat{F}_{[s]_M^*; I}(x),
    \end{equation*}
    where both integrals are Riemann-Stieltjes integrals. By the linearity of the Riemann-Stieltjes integrals, we have
      \begin{align*}
       |V_\text{Call}^{MC}([s]_M^*;K,t) - V_\text{Call}^{\text{Market}}(K,t)| &= \left|\int_{-\infty}^\infty (x-K)^+ \d \hat{F}_{[s]_M^*; I}(x) - \int_{-\infty}^\infty (x-K)^+ \d F_{I(t)}(x) \right|\\
       &= \left|\int_{-\infty}^\infty (x-K)^+ \d(\hat{F}_{[s]_M^*; I} - F_{I(t)})(x)\right|  \\
       &= \left|\int_{a}^b (x-K)^+ \d (\hat{F}_{[s]_M^*; I} - F_{I(t)})(x)\right|.
    \end{align*}  
    Since $\hat{F}_{[s]_M^*; I} - F_{I(t)}$ is zero at the boundaries, integration by parts yields 
          \begin{align*}
       \left|\int_{a}^b (x-K)^+ \d(\hat{F}_{[s]_M^*; I} - F_{I(t)})(x)\right| &\leq  \left|\int_{a}^b \hat{F}_{[s]_M^*; I}(x) - F_{I(t)}(x) \d x\right|\\
       &\leq \int_{a}^b \left| \hat{F}_{[s]_M^*; I}(x) - F_{I(t)}(x) \right| \d x\\&< \epsilon (b-a).
    \end{align*} 
The proof for put options works the same way. 
\end{proof}

\section{Example: Iterative Sort-Mix (ISM) algorithm \& Implementation Details}
\label{sec:algorithm}
 We will now present an implementation of a rearrangement algorithm $\mathbf{A}(M,t; \Sigma, \omega)$ for basket option pricing as defined in the previous section. For any time $t \in T_E$, the algorithm efficiently creates a sample matrix $[s]_M$ which is used to estimate the probability density function $f_{B(t)}$ and therefore allows us to price basket options. Note that, contrary to the previous section, the algorithm is not an optimization, and thus the pricing error cannot be reduced to an arbitrarily low value. The algorithm is nevertheless able to provide sufficiently accurate results. In the empirical section, we will demonstrate that ``sufficient" means that the algorithm performs well enough on empirical data to rapidly and accurately price basket options.

As is usual for a rearrangement algorithm, the first step of the algorithm is to generate the initial sample matrix $[s]_M^\text{init}$. This step is generic, yet a description will be provided in \Cref{sec:sampling}. The core part of the algorithm is the rearrangement of the initial matrix, which we will first explain in a heuristic way. In the following subsections, a detailed procedure is laid out and a pseudo-code example is provided in \ref{app:appendixa}. 

\subsection{Heuristic Explanation}
The core principle of the algorithm is not an optimization as suggested \Cref{eq:algorithm3}. Rather, the algorithm iteratively selects and stores a subset of the samples that are considered suitable, until all samples are selected out. To see how this concretely works, we first notice that there are two general ways to arrange the individual vectors $s_n(t)$: We can sort the array from the smallest sample to the largest sample. In this case, we have
\begin{equation}
    s_n^m(t) \leq s_n^l(t) \iff m \leq l, \hspace{0.5cm} \forall n \leq N.
\end{equation}
Alternatively, we can use a random number generator to arbitrarily mix the individual sample vectors. This is equivalent to applying permutations $\pi_n$ at random and obtaining $\pi(s_n(t))$ for each constituent $n$.

We now also note that in the first arrangement, when all vectors are sorted, the sample correlation between the individual vectors is high. For any pair $n,l \leq N$, we have that the Pearson sample correlation coefficient between $s_n(t)$ and $s_l(t)$ is close to 1. On the other hand, when all sample vectors are randomly ordered, the correlation is close to 0. We now arrange each column in the sample matrix $[s]_M^\text{init}$ in the two arrangements and each time plot the resulting empirical PDF of the index as a histogram. This is given in \Cref{fig:sorted-vs-unsorted}, where we also plot the PDF of the ``target PDF" which is the market-observed PDF of the index.
\begin{figure}[H]
    \centering
    \includegraphics[width=0.7\textwidth]{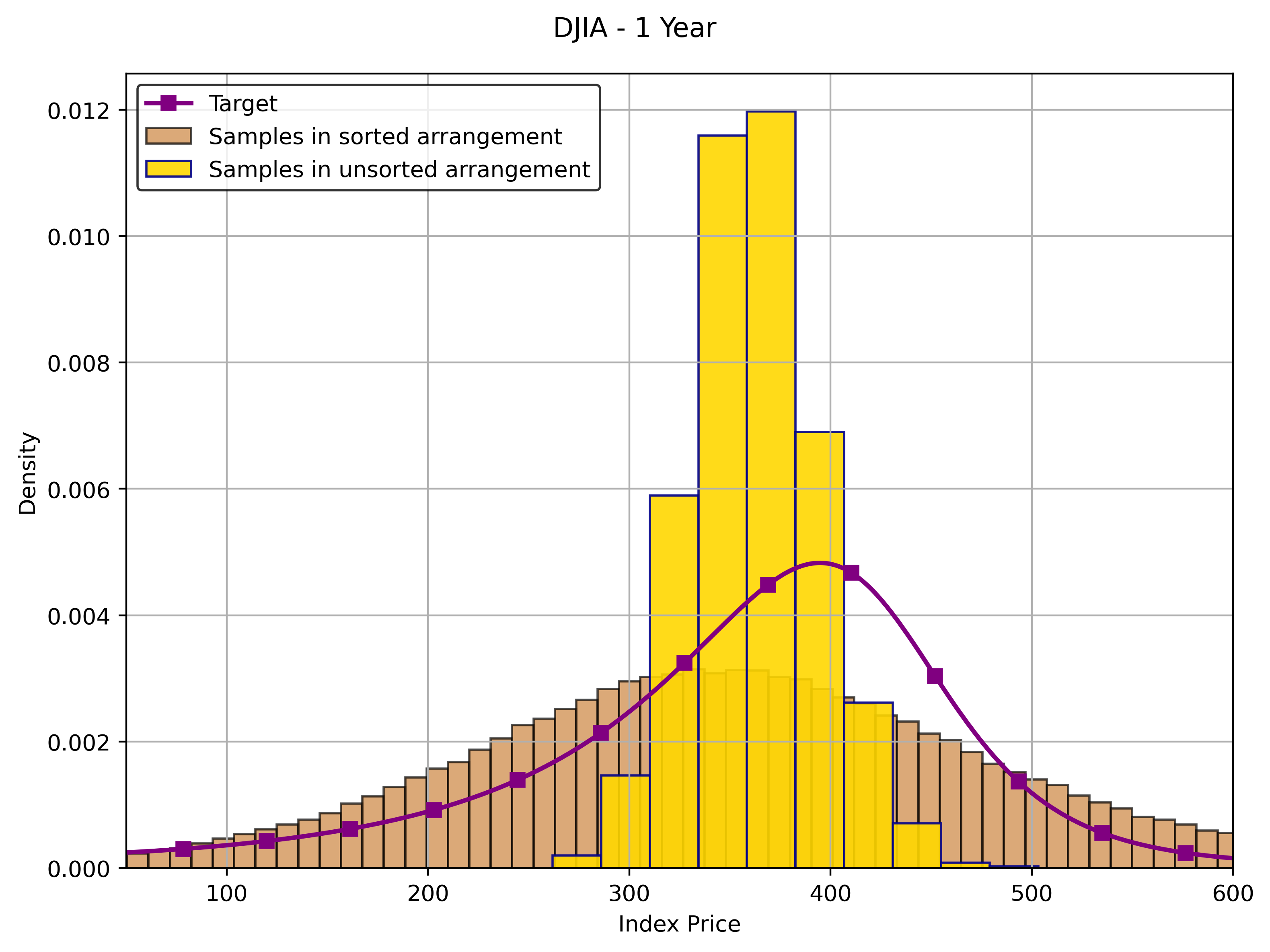}
    \caption{Index probability density function: The graph shows the empirical density of two sample sets of an index as a histogram. The yellow histogram shows the samples when they are in the ``mixed" arrangement. The brown histogram shows the ``sorted" arrangement, displaying the fat tails due to the strong correlation of the constituents - Samples based on data from DJIA 1 Year }
    \label{fig:sorted-vs-unsorted}
\end{figure}
The rationale behind the algorithm is straightforward: We want to create a superposition of sorted samples and mixed samples, such that the empirical probability density function matches the target PDF. To do so, we iteratively arrange the samples in the two arrangements, and store the samples which ``fit under the target PDF". Fitting under the target PDF means that the resulting empirical density function is locally below the target PDF. We then continue the iteration only with the samples that are not stored and choose the other arrangement. The algorithm always starts with the ``sort" arrangement. 

After only a few iterations, the algorithm can no longer make progress, and the remaining samples are added to the previously stored ones. This results in a sub-optimal empirical distribution compared to the target. As we will see, this error is not significant, and the function $f_{B(t)}$ can still be estimated effectively.

The resulting sample matrix $[s]_M^*$ is thus $\epsilon$-admissible up to the error $\epsilon$ of the remaining samples which cannot be fit. Furthermore, the samples in both sorted and mixed form are symmetric, which means that the resulting superposition of samples is also symmetric.

\Cref{fig:4steps} shows a graphical example of one iteration of the algorithm \footnote{the data is from \Cref{sec:marketData}}.
\begin{figure}[H]
    \centering
    \includegraphics[width=\textwidth]{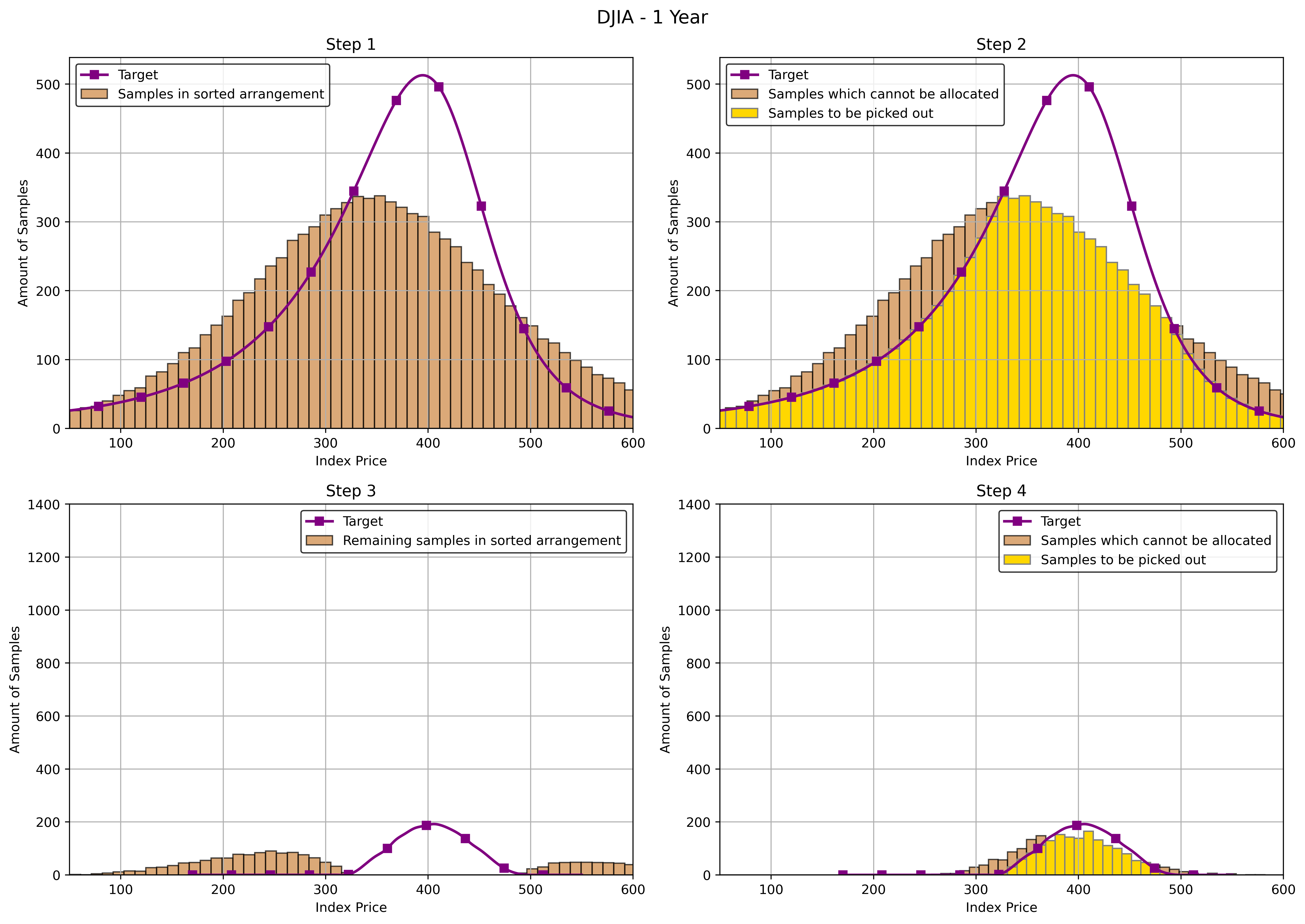}
    \caption{One iteration of sorting - mixing procedure: The algorithm starts with 10'000 samples. Step 1 (Top Left): Samples are ordered. Step 2 (Top Right): All valid samples (below target PDF) are removed from the sample matrix and stored away. Step 3 (Bottom Left): The target is updated, and the remaining samples are outside of the target range. Step 4 (Bottom Right): The remaining samples are mixed together and are now ready to be removed again. Only $\sim$ 400 samples are left after this iteration.}
    \label{fig:4steps}
\end{figure}
\subsection{Independent Constituent Sampling}
\label{sec:sampling}
The initial step consists of generating $M \in \N$ samples of the constituent prices according to their marginal distribution, to initiate the matrix $[s]_M^\text{init}$. This is achieved by generating uniform random variables and composing them with the inverse of the cumulative distribution function of the constituents. Let $t >0$ be some expiry time, and let $M$ be the number of required samples. Given the probability distributions $F_{S_n(t)}$ of $S_n(t)$, which we obtain from calibrating \Cref{eq:lvm} to the constituents' vanilla option prices, we generate the samples separately for each constituent. We will use the following mapping: 
\begin{equation}
\label{eq:estimation}
    S_n(t) \stackrel{d}{=} F_{S_n(t)}^{-1}\left(U_n\right),
\end{equation}
where $U_n, 1  \leq n \leq N$ are distinct and independent uniformly distributed random variables. Therefore, we obtain the samples $s_n(t)$ as
\begin{equation}
    s_n(t) = \left[s^1_n(t), s^2_n(t),\dots,s^M_n(t) \right]^T = \left[F^{-1}_{S_n(t)}(u_n^1), F^{-1}_{S_n(t)}(u^2_n),\dots,F^{-1}_{S_n(t)}(u^M_n )\right]^T ,
\end{equation}
where $u_n^1,u_n^1,\dots,u_n^M $ are the uniform samples.

\subsection{Range Discretization and Target Vector}
Before the algorithm can arrange the samples, we require a discretization of the involved PDFs. We aim to obtain a target vector $V \in \N^{k}$ with $k \in \N$, which specifies the expected number of samples in a subset of the range of $I(t)$ under the target PDF $f_{I(t)}$. 

We begin by defining a discretization of the range of the index. Since the range of $I(t)$ is $\R_+$ and therefore unbounded, we need to reduce it to an appropriately bounded set $[g_0,g_K] \subset \R_+$ to ensure numerical efficiency. We choose the bounds $g_0< g_K$ for the interval with the following interpretation: For $M$ samples from an admissible copula, we expect exactly one sample below $g_0$ and one sample above $g_k$. These values are given by the probability distribution function \begin{equation}
    g_0 = F^{-1}_{I(t)}\left(\frac{1}{M}\right)\hspace{0.5cm} \text{ and }\hspace{0.5cm} g_K = F^{-1}_{I(t)} \left(\frac{M-1}{M}\right).
\end{equation}
We then define an equidistant grid $g_0 < g_1 < g_2 < \dots < g_K$, which define $K$ bins $G_k = (g_{k-1}, g_k]$\footnote{For completeness, we extend $G_0$ with $g_0$ to $[g_0,g_1]$} to partition the interval \footnote{The parameter $K$ is a factor for the quality and computational complexity of the algorithm and should be chosen in relation to $M$, for instance $\frac{M}{K}=10$}. We thus have
\begin{equation}
    [g_0,g_K] = \bigcup_{0 \leq k \leq K} G_k.
\end{equation}
For each bin, we can now calculate the expected samples to fall into each bin under $F_{I(t)}$. This value is given by the cumulative distribution, rounded to an integer
\begin{equation}
    \Delta F_k =  \big \lceil M \;\Q( I(t) \in G_k)\big\rfloor =\big\lceil M \left(F_{I(t)}(g_k) - F_{I(t)}(g_{k-1})\right) \big\rfloor \in \N,
\end{equation}
where $ \lceil x \rfloor$ denotes $x$ rounded to its nearest integer. We capture the discrete PDF with the \emph{target vector} 
\begin{equation}
\label{eq:targetVector}
    V = \left[\Delta F_1, \Delta F_2 ,\dots,\Delta F_K\right] \in \N^K.
\end{equation}
Furthermore, we introduce a discretization of the loss function on the PDF, given by $\hat{\ell}$. For a sample matrix $[s]_M \in \M$, we define the bin count $c([s]_M)$ with
\begin{equation}
c_k([s]_M) = |\{1 \leq m \leq M : i^m \in G_k\}|,\end{equation}
with $1\leq k \leq K$. The discrete error is now defined as the difference between expected and actual discrete PDF:
\begin{equation}
    \hat{\ell}([s]_M) = \frac{\sum_{k=1}^K|c_k([s]_M) - \Delta F_k |}{2M} \in [0,1].
\end{equation}
We divide by two since a sample that cannot be matched to a bin will be counted once twice in the numerator. 
\subsection{Sorting - Mixing Procedure}
With the target vector $V$ defined and samples $s_n(t)$ initiated, we will now describe the details of the core of the algorithm, where the different sample arrangements are iterated. 
The main algorithm thus runs through iterations of first sorting and then mixing samples. Suppose we have a sample matrix $[s]_M$ with constituent sample vectors $s_n(t)$ for all $n \leq N$ and index vector $i(t) = \sum_{n=1}^N s_n(t)$. An iteration of the sorting-mixing procedure starts by arranging the constituent sample vectors in sorted order and obtaining the index sample vector $i(t)$. Now, for every bin $G_k$, we collect the vector positions $m$ \footnote{By ''vector positions" we mean the index $m$ of 
$i^m(t)$ of the vector $i(t)$. We refrain from calling it the ``index" of the vector due to the ambiguity with the index $I(t)$. } of the samples that have values $i^m(t)$ in $G_k$. We define the set of positions of the samples in $G_k$ as
\begin{equation}
    R_k = \{1 \leq m \leq M : i^m(t) \in G_k \}.
\end{equation}
The sample count $c_k$, which was defined previously, is thus given by $c_k = |R_k|$. We want to determine a set of positions $\Bar{R_k} \subset R_k$, which we can store away and remove from the sample matrix. The count $c_k$ does not match the expected count, given by $\Delta F_k$. We therefore want to remove \emph{at most} $\Delta F_k$ samples and we might have to reduce the set $R_k$ of samples to remove. If $c_k >  \Delta F_k$, we need to reduce the set $R_k$ to a smaller set by $\Bar{R}_k$. We do this by randomly selecting $\Delta F_k$ elements from $R_k$. If $c_k \leq  \Delta F_k$, all the samples from $R_k$ can be removed and we have $\Bar{R}_k = R_k$. In both cases, we now have
\begin{equation}
    \Bar{c}_k = |\Bar{R}_k|= \min(c_k, \Delta F_k).
\end{equation}
To visualize the involved quantities, we see that the height of the brown histogram in \Cref{fig:4steps} is given by $c_k$, while the height of the yellow columns is $\Bar{c}_k$. 

We can now store all samples with indices in $\Bar{R}_k$ and reduce the sample vectors. The updated constituent sample vectors $s_n(t)'$ contain all elements $ s_n^m(t) $ such that $ m \notin \bigcup_{k \leq K} \Bar{R_k}$. The length of $s_n(t)'$ is thus $M - \sum_{k \leq K} \Bar{c}_k$ for all constituents.

Lastly, the target vector needs to be updated too, since we no longer require that many samples per bin. Since we remove $\Bar{c}_k$ samples, we update as
\begin{equation}
    \Delta F_k' = \Delta F_k - \Bar{c}_k.
\end{equation}
To finish one iteration of the algorithm, we need to repeat this procedure in a random arrangement. We thus now ``shuffle'' the elements of $s_n(t)'$ and then repeat the derivation of $R_k,\Bar{R}_k$, et cetera, removing and storing more samples. After sorting and shuffling once, we conclude one iteration.  \Cref{fig:result1y} shows how the final PDF of the samples matches the target PDF.

\begin{figure}[H]
    \centering
    \includegraphics[width=\textwidth]{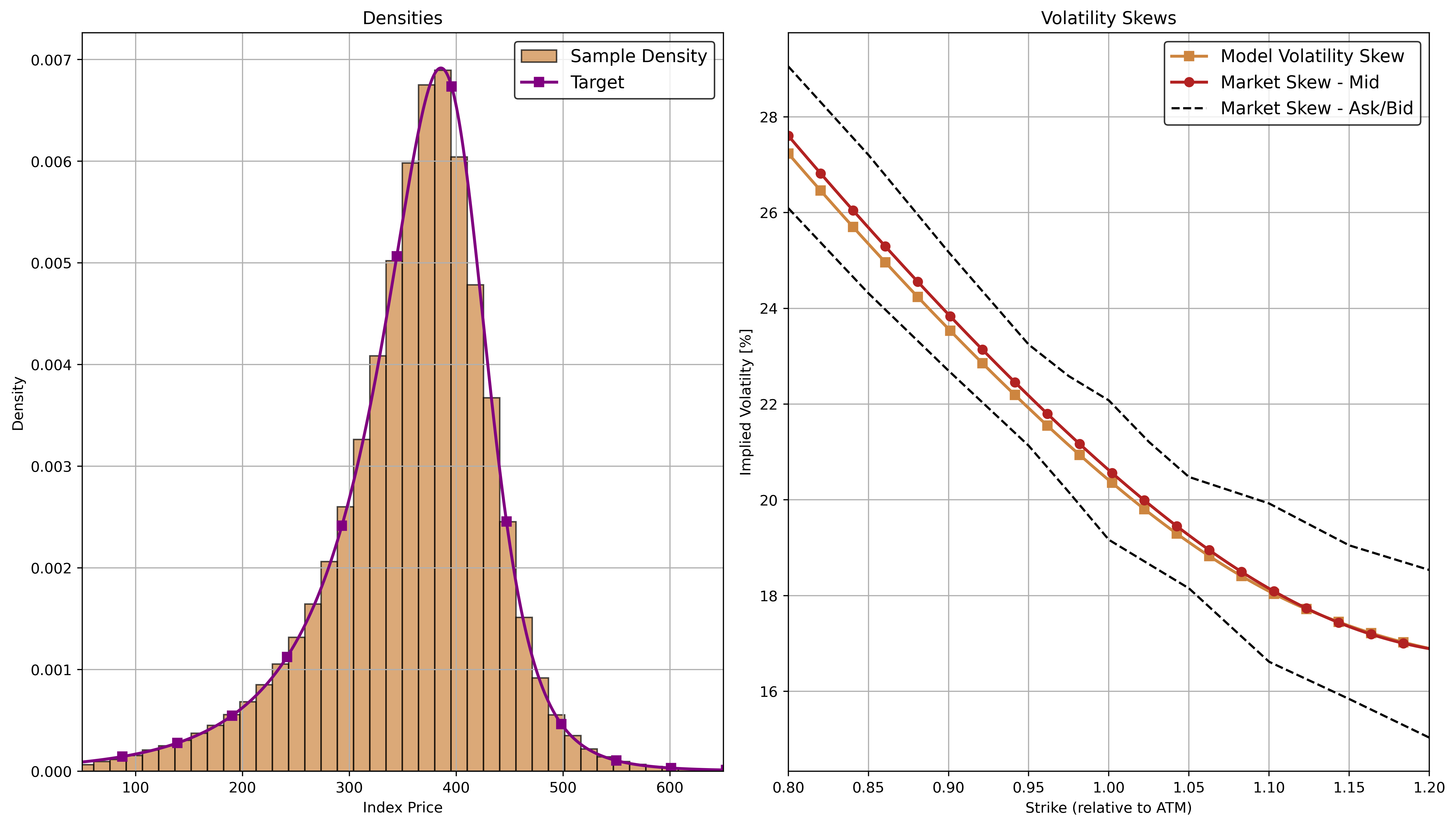}
    \caption{Final PDF after 5 iterations. The match of the empirical PDF and target PDF means that options are priced correctly. This can be seen on the implied volatility graph on the right (Example DJIV 1Y)}
    \label{fig:result1y}
\end{figure}
\begin{example}
We demonstrate the workings of the algorithm on a small dataset with 2 assets and 10 samples. The samples are marginally uniformly distributed on $[0,1]$. The index samples fall into $[0,2]$, which we split into $K=3$ equally large bins with boundaries $$ \left\{0,\frac{2}{3},\frac{4}{3},2 \right\}.$$
The bins are setup such that\footnote{Since the $[0,2]$ is already bounded, there is no need to obtain $g_0$ and $g_K$ by the quantiles. }
\begin{equation*}
    G_1 = \left[0,\frac{2}{3}\right], \quad G_2 = \left(\frac{2}{3},\frac{4}{3}\right], \quad G_3 = \left(\frac{4}{3},2\right].
\end{equation*}
Furthermore, suppose we obtain the from the index option quotes the following discrete probability density function for $I$:
\begin{equation*}
    f_{I}(x) = \begin{cases}
    0.3, \text{ if } x \in G_1,\\
    0.5, \text{ if } x \in G_2,\\
    0.2, \text{ if } x \in G_3.
    \end{cases}
\end{equation*}
We then calculate the target vector for $M=10$ samples as 
\begin{equation*}
    \Delta F_k = \lceil M \cdot \Q(I \in G_k) \rfloor, 
\quad k \in \{1,2,3\}
\end{equation*}
from which we obtain
\[V = \left[\Delta F_1,\Delta F_2,\Delta F_3 \right] = \left[3,5,2\right]. \]
\textbf{Algorithm:} We start the algorithm by creating samples for $S_1$ and $S_2$ independently. This step is shown in a) in \Cref{fig:exampleToyData}. Then, in b) we sort both column vectors smallest to largest and sum the columns to obtain the index $I$. For this arrangement we obtain 
 \[c = [c_1,c_2,c_3] = [4,3,3],\] which means that the amount of samples to be removed are
\[
     \bar c = \min(c,\Delta F) = [3,3,2].
\]
 We therefore select three samples from bin 1, three samples from bin 2 and two samples from bin 3, which is shown in c), indicated by the green color that a sample is picked out.  For bin 1 and 3 we select them at random since $c_{1/3} > \Delta F_{1/3}$. Since two samples cannot be selected, the loss at this stage is given by $\ell([s]^*_M) = \frac{2}{10}$. We update the target vectors according to $ \Delta F_k' = \Delta F_k - \Bar{c}_k$ to obtain 
 \[\Delta F_k' = [0,2,0].\]
 The remaining samples are now entering the mixing stage. In d), we mix the samples for $S_1$ and $S_2$ column-wise and aggregate to obtain an new value for $I$. Based on these values we obtain the new $c$ as $c = [0,2,0]$. Since this is exactly the remaining target vector, we select both new samples and the algorithm is finished. The remaining error is thus $\ell([s]_M^*) = 0$ and the selected samples can be used.
\begin{table}[H]
\centering
  \begin{subfigure}{0.45\textwidth}
  \centering
\begin{tabular}{|c c|}
    \hline
    \rowcolor{myblue!30} {$S_1$} & {$S_2$} \\
    \hline
    \hline
    0.29 & 0.31 \\
    0.14 & 0.47 \\
    0.26 & 0.17 \\
    0.44 & 0.07 \\
    0.05 & 0.01 \\
    1.00 & 0.69 \\
    0.31 & 0.83 \\
    0.76 & 0.49 \\
    0.72 & 0.41 \\
    0.04 & 0.76 \\
    \hline
  \end{tabular}
    \caption{Initial data sampled as two independent uniform random variables}
  \end{subfigure}
  \begin{subfigure}{0.45\textwidth}
    \centering
\begin{tabular}{ |c| c c | c |}
\hline
   \rowcolor{myblue!30} &$S_1$ & $S_2$ & $I$ \\
\hline
    \hline
\cellcolor{myblue!10} & 0.04 & 0.01 & 0.05 \\

  \cellcolor{myblue!10} & 0.05 & 0.07 & 0.12 \\

   \cellcolor{myblue!10}& 0.14 & 0.17 & 0.31 \\

   \multirow{-4}{*}{\cellcolor{myblue!10}\rotatebox[origin=c]{90}{$G_1$}}& 0.26 & 0.31 & 0.57 \\
  \hline
  \cellcolor{myblue!10} & 0.29 & 0.41 & 0.70 \\

   \cellcolor{myblue!10}  & 0.31 & 0.47 & 0.78 \\

    \multirow{-3}{*}{\cellcolor{myblue!10}\rotatebox[origin=c]{90}{$G_2$}}& 0.44 & 0.69 & 1.13 \\
  \hline
 \cellcolor{myblue!10}& 0.72 & 0.76 & 1.48 \\

    \cellcolor{myblue!10} & 0.76 & 0.83 & 1.59 \\

  \multirow{-3}{*}{\cellcolor{myblue!10}\rotatebox[origin=c]{90}{$G_3$}}& 1.00 & 0.76 & 1.76 \\
  \hline
\end{tabular}
    \caption{Sort data column-wise, compute $I = S_1 + S_2$ and split into bins based on index value}
  \end{subfigure}

  \bigskip
    \begin{subfigure}{0.45\textwidth}
    \centering
  \begin{tabular}{|c | c c | c| }
    \hline
    \rowcolor{myblue!30}&{$S_1$} & {$S_2$} & {$I$ }  \\
    \hline
    \hline
   \cellcolor{myblue!10}  &    \cellcolor{green!15}0.04 & \cellcolor{green!15}0.01&\cellcolor{green!15}0.05  \\
      
   \cellcolor{myblue!10}  &\cellcolor{green!15}0.05 & \cellcolor{green!15}0.07 &\cellcolor{green!15}0.12 \\
    \cellcolor{myblue!10} &\cellcolor{white!0}0.14 & 0.17 &\cellcolor{white!0}0.31 \\
    \multirow{-4}{*}{\cellcolor{myblue!10}\rotatebox[origin=c]{90}{$G_1$}}&\cellcolor{green!15}0.26 & \cellcolor{green!15}0.31 & \cellcolor{green!15}0.57 \\
    \hline 
     \cellcolor{myblue!10}&\cellcolor{green!15}0.29 &\cellcolor{green!15} 0.41 & \cellcolor{green!15}0.70\\
     \cellcolor{myblue!10}&\cellcolor{green!15}0.31 & \cellcolor{green!15}0.47 & \cellcolor{green!15}0.78\\
     \multirow{-3}{*}{\cellcolor{myblue!10}\rotatebox[origin=c]{90}{$G_2$}}&\cellcolor{green!15}0.44 & \cellcolor{green!15}0.69  &\cellcolor{green!15}1.13 \\
    \hline
     \cellcolor{myblue!10}&\cellcolor{green!15}0.72 & \cellcolor{green!15}0.76 &\cellcolor{green!15}1.48 \\
     \cellcolor{myblue!10}&\cellcolor{white!0}0.76 & 0.83 &\cellcolor{white!0}1.59 \\
     \multirow{-3}{*}{\cellcolor{myblue!10}\rotatebox[origin=c]{90}{$G_3$}}&\cellcolor{green!15}1.00 & \cellcolor{green!15}0.76 & \cellcolor{green!15}1.76 \\
    \hline
  \end{tabular}
    \caption{Select $\bar c_k$ samples per $G_k$. Two samples cannot be attributed to the right bin.}
  \end{subfigure}
   \begin{subfigure}{0.45\textwidth}
    \centering
  \begin{tabular}{|c |c c | c|}
    \hline
  \rowcolor{myblue!30} & {$S_1$} & {$S_2$} & {$I$} \\
    \hline
    \hline

   \cellcolor{myblue!10} & \cellcolor{green!15} 0.14 & \cellcolor{green!15} 0.83 &\cellcolor{green!15} 0.97 \\
  \multirow{-2}{*}{\cellcolor{myblue!10}\rotatebox[origin=c]{90}{$G_2$}} & \cellcolor{green!15} 0.76 &\cellcolor{green!15} 0.17 &\cellcolor{green!15} 0.93 \\
    \hline
  \end{tabular}
    \caption{The remaining two samples are mixed and binned.}
  \end{subfigure}
  \caption{One iteration on toy data}
  \label{fig:exampleToyData}
\end{table}
    
\end{example} 

\section{Numerical Results}
\label{sec:marketData}
\subsection{Setup}
It remains to demonstrate the effectiveness of the ISM algorithm on actual market data. In this section, we will consider an experiment on actual market data for the Dow Jones Industrial Average (DJIA) and its 30 constituents as of 12th August 2021\footnote{An accompanying Python repository of the implementation is available at \url{https://github.com/NFZaugg/BasketOptionsRearrangement}}. The market data consists of implied volatility surfaces and spot prices for the 30 assets and the index. \Cref{fig:data} shows an overview of the available data and the selection for the parameters.

To extract the marginal distribution of the constituents as well the index distribution we calibrate SABR (Stochastic Alpha-Beta-Rho) parameters \cite{HagaKumaLesnWood02} for each of the 30 constituents and each maturity $T$, fixing parameter $\beta$ to $0.9$. The choice of SABR is not relevant as long as it provides a smooth price of $S_n(t)$, such that we can estimate $F_{S_n}(t)$. The first few of the calibrated parameters are shown in \Cref{fig:head-param}. The same model is calibrated for the index options which enables the estimation of $F_{I^\text{Mkt}(t)}$.
\begin{table}[H]
  \centering
\begin{tabular}{>{\columncolor{myblue!10}}l| >{\columncolor{myblue!10}}l}
    \hline
    \multicolumn{2}{c}{Data Information} \\
    \hline
        \hline
        Date & 2021/08/12\\
        Number of constituents $(N)$ & 30\\
        Constituent weights & 0.066 \\
    Vol Surface Moneyness & 0.8, 0.85, 0.9, 0.95, 0.975, 1, 1.025, 1.05, 1.1, 1.15, 1.2\\
    Vol Surface Maturity &  3m, 6m, 1y, 1.25y, 2y \\
    \hline
  \end{tabular}
  \newline\newline\newline
    \centering
\begin{tabular}{>{\columncolor{myblue!10}}l| >{\columncolor{myblue!10}}l}
    \hline
    \multicolumn{2}{c}{Hyperparameters} \\
    \hline
        \hline
        Number of Samples $(M)$ & 20'000\\
        Number of Bins & 1'400\\
        Number of Iterations & 10\\
        Constituents Marginal Model & SABR, fixed $\beta =0.9$\\
        Index model & SABR, fixed $\beta =0.9$\\
    \hline
  \end{tabular}
      \caption{Experiment setup}
      \label{fig:data}
\end{table}

We generate a matrix of $M$ samples of $N$ independent uniform random variables and use \Cref{eq:estimation} to transpose those into random variables with the proper marginal distribution function. 

With the distribution of the index, we calculate the target vector (\ref{eq:targetVector}) with 2'000 equidistant bins, yielding the discrete PDF we are aiming to match. Having collected all the input, we run the algorithm to rearrange the samples.
\begin{table}[H]
    \centering
    \begin{tabular}{>{\columncolor{myblue!10}}l|>{\columncolor{myblue!10}}c >{\columncolor{myblue!10}}c >{\columncolor{myblue!10}}c >{\columncolor{myblue!10}}c}
    \hline
        \cellcolor{white}Name & \cellcolor{white}$\beta$ & \cellcolor{white}$\alpha$ &\cellcolor{white} $\rho$ & \cellcolor{white}$\gamma$ \\
        \hline
        \hline
        United Health & 0.9 & 0.47 & -0.52 & 1.20 \\
        Home Depot & 0.9 & 0.45 & -0.28 & 1.47 \\
        Goldman Sachs & 0.9 & 0.50 & -0.33 & 1.46 \\
        Microsoft Corp & 0.9 & 0.43 & -0.43 & 1.49 \\
        \dots & \dots &&\dots&\\
        \hline
        DJIA & 0.9 & 0.31 & -0.64 & 1.91 \\
           \hline
    \end{tabular}
    \caption{First 4 calibrated parameters, $T=3m$}
    \label{fig:head-param}
\end{table}
\subsection{Index Repricing}
\Cref{fig:marketvsmodel} shows the implied volatility of the index option market prices vs the prices obtained by the model. The model can reprice the index option from the market well. The remaining error is split between calibration inaccuracies of $F_{I(t)}$ and the discrete error, which is displayed in \Cref{fig:sampleless}. The errors, which do not exceed 0.2\% in the implied volatilities per moneyness, are well within the bid and ask quotes of the options (ranging from 0.2\% to 4\%). We conclude that the model is fit to price basket options according to the market-implied correlation skew.
\begin{figure}[H]
    \centering
    
    \includegraphics[width=\textwidth]{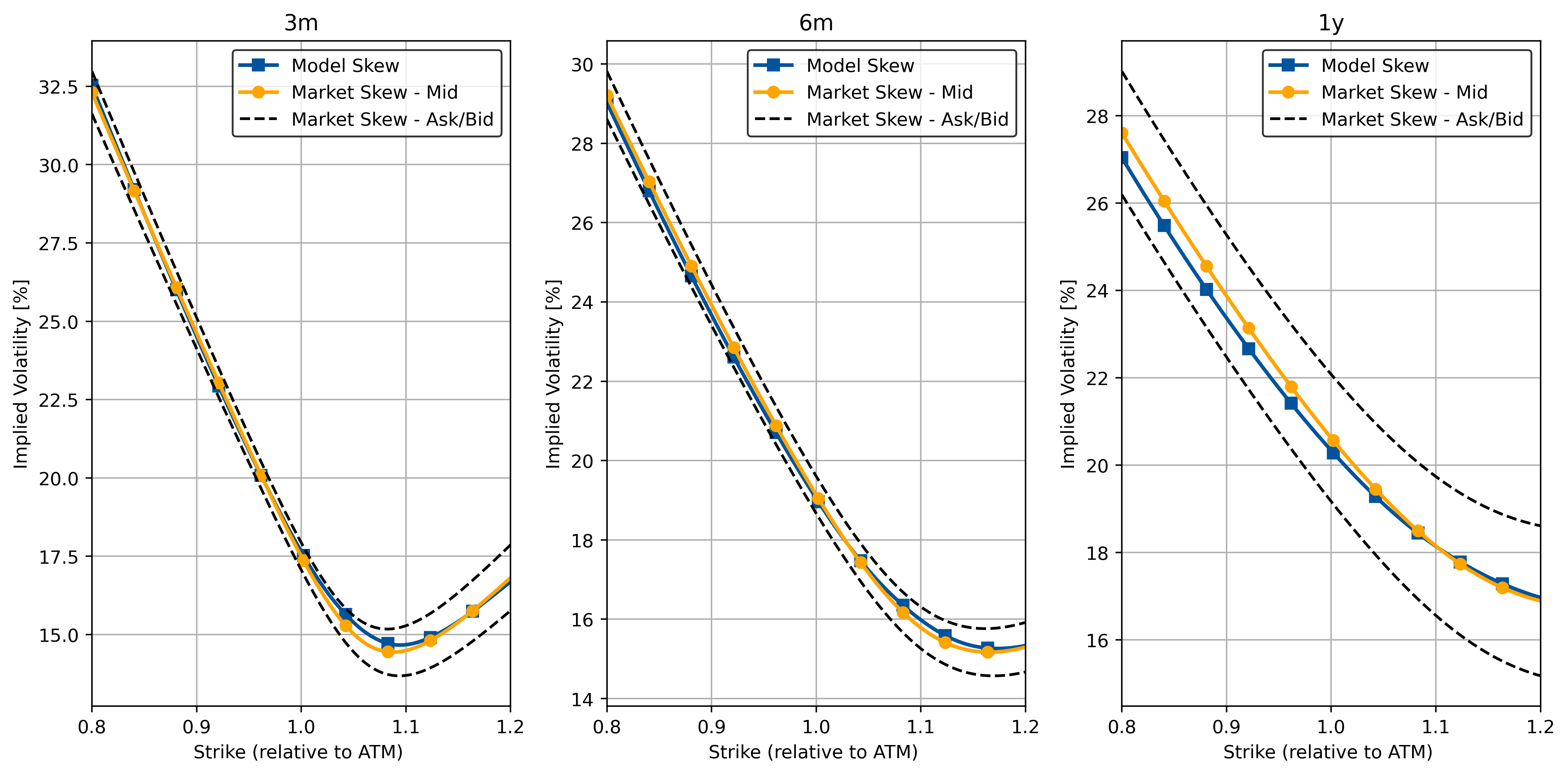}
    \includegraphics[width=\textwidth]{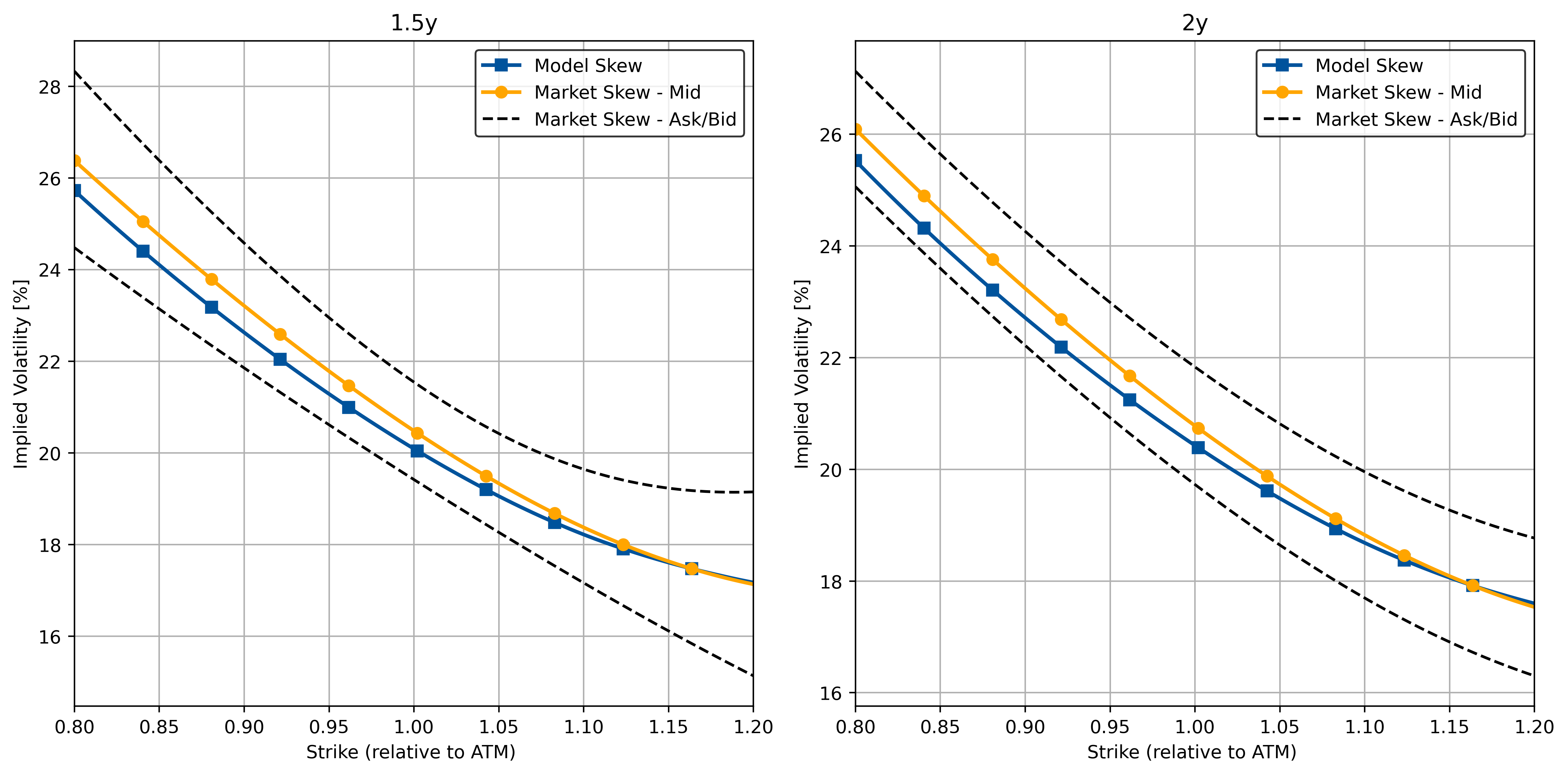}
    \caption{Market vs model implied volatilities - DJIA options. Note: The increased mispricing of 1.5y and 2y mainly stem from numerical inaccuracies in $F_{I(t)}$, which assumes significant mass at $I(t) = 0$ }
    \label{fig:marketvsmodel}
\end{figure}

\begin{table}[H]
    \centering
    
  \begin{tabular}{>{\columncolor{myblue!10}}c>{\columncolor{myblue!10}}c>{\columncolor{myblue!10}}c>{\columncolor{myblue!10}}c>{\columncolor{myblue!10}}c>{\columncolor{myblue!10}}c}
\hline
    \rowcolor{white} Maturity & 3m & 6m & 1y & 1.5y&2y \\
        \hline
    Discrete Error $\hat{\ell}(S) $& $2.09\%$ & $1.54\%$ & $1.6\%$& $1.92\%$ & $2.27\%$ 
  \end{tabular}
      \caption{Discrete error per maturity}
      \label{fig:sampleless}
\end{table}
\subsection{Performance}
Given the results from the previous subsection, we analyze the time spent on each step in the calibration of the model. \Cref{fig:performance} shows that for each maturity, the total time spent was between 5 and 7 seconds. Most of the time was spent on calibrating constituents and generating independent samples. The arranging of the samples took about 1/3 of the time. Given that this algorithm can be run in parallel, the LVM can thus be calibrated in below 7 seconds with 20'000 samples.
\begin{table}[H]
    \centering
    
  \begin{tabular}{>{\columncolor{myblue!10}}l>{\columncolor{myblue!10}}c>{\columncolor{myblue!10}}c>{\columncolor{myblue!10}}c>{\columncolor{myblue!10}}c>{\columncolor{myblue!10}}c}
\hline
    \rowcolor{white} Action & 3m & 6m & 1y & 1.5y&2y \\

    Calibrating Constituents& 1.76 & 1.60 & 1.76& 1.83 &2.09\\
     \hline
    Calibrating Index (Bid Ask Mid)& 0.19 & 0.19 & 0.2& 0.2 &0.19\\
     \hline
    Sampling Constituents& 1.44 & 1.52 & 1.63& 1.56 &1.79\\
     \hline
    Preparing Target & 0.07 & 0.08 &0.08& 0.07 &0.08\\
     \hline
    Arranging Samples& 1.27 & 1.56& 1.72&1.92&2.1\\
     \hline
      \hline
      Total & 5.73s & 5.95s & 6.39s & 6.58s & 6.25s\\
  \end{tabular}
      \caption{Time Breakdown (s) for example of 30 constituents and 20'000 samples}
      \label{fig:performance}
\end{table}
\subsection{Greeks Calculation}
The rapid calibration of the model allows for fast Greeks calculation for the basket derivatives with a finite difference approximation. Suppose that $\Theta_{LVM} = \Theta_{LVM}(\sigma, \sigma_I)$ is the local variance function of (\ref{eq:lvariance}) after calibration given the input market data $\sigma$ and $\sigma_I$. Here, $\sigma = (\sigma_1,\sigma_2,\dots,\sigma_N)$ and $\sigma_I$ are the implied volatility curves of the constituent and index, respectively.

We express the price of a derivative $V\left(B(0); \Theta_{LVM}\right) = V\left(B(0)\right)$ as a pricing function of a basket derivative given the parameters and a spot price $B(0)$. The delta and gamma of $V$, the first and second derivatives with respect to the spot price, are approximated using finite differences, for instance
\begin{align}
    \Delta_B = \frac{\partial  V\left(B(0)\right)}{\partial B(0) } &\approx   \frac{V\left(B(0) + \epsilon\right) - V\left(B(0)\right) }{\epsilon}, \\
        \gamma_B =\frac{\partial^2  V\left(B(0)\right)}{\partial B(0) ^2} &\approx \frac{V\left(B(0) - \epsilon\right) - 2 V\left(B(0)\right) + V\left(B(0) + \epsilon \right) }{\epsilon^2},
\end{align}
    for some bump size $\epsilon > 0$. The sensitivity to the individual constituents is given by the chain rule \footnote{Note that if the basket is weighted, the delta needs to be weighted too}
\begin{align}
    \frac{\partial  V\left(B(0)\right)}{\partial S_n(0) } &= \Delta_B \cdot \frac{\partial B(0)}{\partial S_n(0)} =  \Delta_B.
\end{align}
Furthermore, one can also calculate the sensitivities to the constituent implied volatility curves. Let vega $v_n$ be the sensitivity of the basket derivative to the $n$-th constituent implied volatility curve,
\begin{align}
    v_n =  \frac{\partial  V\left(B(0)\right)}{\partial \sigma_n } \approx \frac{V\left(B(0); \Theta_{LVM}(\sigma+\mathbb{\epsilon}_n, \sigma_I)\right) - V\left(B(0); \Theta_{LVM}(\sigma, \sigma_I)\right)}{\epsilon},
\end{align}
where $\mathbb{\epsilon}_n$ is the vector $(0,0,\dots,0,\epsilon,0,\dots,0)$. To obtain $V\left(B(0); \Theta_{LVM}(\sigma_n+\epsilon, \sigma_I)\right)$ we need to re-calibrate the model, as $\Theta_{LVM}$ are changing due to the shift in implied volatility. Fortunately, the re-calibration is a surprisingly computationally cheap procedure. As a first step, we re-calibrate the volatility model for the constituent $n$. This provides us with an updated probability distribution function 
\begin{equation}
    F_{S_n(t)}'(x) = \Q \left( S_n(t; \sigma_n+\epsilon) < x )\right), \; x \in \R ,
\end{equation} where $S_n(t; \sigma_n+\epsilon)$ indicates the constituent price process at $t$ under its new distribution due to the shift in volatility. To maintain the same dependency structure as before the shift, we have to adjust the samples consistently. In the initial calibration, for a fixed $t \in T_E$, we obtained a sample vector for a uniform random variable $(u_n^1, u_n^1,\dots, u_n^M)$. Furthermore, we used a permutation $\pi_n$ to rearrange the samples. We can now simply obtain the samples with the new distribution as
\[ \{{F_{S_n(t)}'}^{-1} (u_n^m) : m \leq M)\},\]
and order them in the same way as the old samples
\[ \{{F_{S_n(t)}'}^{-1} (\pi_n(u_n^m)) : m \leq M\}.\]
Important is that we reuse the same samples for the uniform random variable $u_n^m$, as the permutation is determined based on these values. Alternatively, the same can be obtained by applying the function \begin{equation}
    x \mapsto {F_{S_n(t)}'}^{-1} \left( F_{S_n(t)}\left(x\right)\right) ,
\end{equation} to the pre-shift samples of $S_n(t)$.
We thus only adjust the samples for the $n$-th constituent and re-calibrate the local variance function for every $t\in T_E$. The total computational cost is low as no new sample ordering is required. Results for $\Delta,\gamma,v$ are shown in \Cref{fig:deltaGamma}, \Cref{fig:vega}. Other Greeks are obtained similarly.
\begin{table}[H]
    \centering
    
  \begin{tabular}{>{\columncolor{myblue!10}}l>{\columncolor{myblue!10}}c>{\columncolor{myblue!10}}c>{\columncolor{myblue!10}}c>{\columncolor{myblue!10}}c>{\columncolor{myblue!10}}c}
\hline
    \rowcolor{white} Strike (relative to ATM) & 0.8& 0.9&1&1.1&1.2 \\

    $\Delta$& 0.95 & 0.88 & 0.60& 0.09 &0.01\\
    $\gamma$& 1e-7 & 1e-6 & 0.03& 1e-10 & 1e-10\\
  \end{tabular}
      \caption{$\Delta, \gamma$ of DJIA 3m European option at spot price for various strikes }
      \label{fig:deltaGamma}
\end{table}
\begin{table}[H]
    \centering
    
    \begin{tabular}{>{\columncolor{myblue!10}}l|>{\columncolor{myblue!10}}c }
    \hline
        \cellcolor{white} Constituent Name & \cellcolor{white}$v_i$ \\
        \hline
        \hline
        United Health & 4.63\%  \\
        Home Depot & 3.68\%  \\
        Goldman Sachs & 3.64\% \\
        Microsoft Corp & 3.02\%  \\
        Salesforce Inc & 2.32 \%\\
           \hline
    \end{tabular}
      \caption{Vega of DJIA 3m European option for first 5 constituents at the money (Upwards finite difference, $\epsilon = 0.01$  (1\% in implied volatility)) }
      \label{fig:vega}
\end{table}
\section{Conclusions}
\label{sec:concl}
Pricing basket options consistently with index options poses a significant challenge. Creating a model that accurately reproduces the index skew is often challenging to calibrate efficiently and, therefore, undesirable for practical purposes. The difficulty arises from the intricate nature of modeling the dependency structure between constituents, which is crucial for accurate pricing and is often not straightforward to extract from available market data.

Copula functions, describing the joint distribution of random variables mapped to a uniform space, emerge as valuable tools for modeling such dependency structures. Calibrating a copula model to the available market data remains difficult, due to two main issues. Firstly, the inherent multidimensionality of the problem makes calibration of a multidimensional model computationally expensive. Furthermore, the problem of underdetermination of the copula $C$ can lead to errors in the pricing of basket options, even if the model can replicate index options without errors. In this paper, we define a copula $C$ to be a valid pricing model for basket options if the copula is symmetric and replicates the index PDF.

Recognizing the challenges of directly modeling copulas, we utilize rearrangement algorithm to imply an admissible and symmetric copula for the dependency of the constituents. These algorithms work by initializing samples according to the known marginal distribution and then pragmatically rearranging them to reflect the dependency structure to match the index probability distribution. Extra constraints based on the empirical symmetry are added to infer a symmetric copula.

We develop a particular rearrangement algorithm, called the Iterative-Sort-Mix algorithm, that automatically calibrates the local volatility function using a simple rearrangement technique, sidestepping the need for an expensive multivariate parameter optimization. Furthermore, the algorithm calculates the Greeks, in particular sensitivities to the constituents' implied vol curves, with ease and avoids additional expensive calculations.

The viability of this approach is tested on historical market data, revealing its effectiveness in calibrating a model for pricing baskets with up to 30 constituents within seconds and with limited errors. The algorithm thus exhibits good performance and is suitable for practical purposes. 

Furthermore, based on the theoretical considerations, the approach encourages the development of novel and enhanced rearrangement algorithms with robust convergence properties. A promising direction is the use of artificial intelligence to address the high-dimensional problem of rearranging the samples in an appropriate order. 
\bibliography{bib}
\appendix
\section{Algorithm Pseudo-Code}
\label{app:appendixa}
We express the algorithm as a pseudo-code:
\newline\newline
\begin{algorithm}[H]
\SetAlgoLined
    \SetAlgoNlRelativeSize{-1}
    \SetAlgoNlRelativeSize{-1}
    \SetKwFunction{FSort}{Sort}
    \SetKwFunction{FSelect}{SelectValidSamples}
    \SetKwFunction{FMix}{Mix}
        \SetKwFunction{FMain}{Main}
    \caption{Iterative Sort-Mix}
    \tcc{As an input, we have initialized samples, the discretized target vector, and the corresponding bins of the discretization}
    \KwInput{sampleMatrix: $(M \times N)$, targetVector $(g\times 1)$, bins $((g+1) \times 1)$}
    
    \KwOutput{outputMatrix: $(M \times N)$}
    \vspace{0.3cm}
          \SetKwProg{Fn}{Def}{:}{}
    \Fn{\FMain{sampleMatrix,targetVector,bins}}{
    remainingMatrix $\leftarrow$ sampleMatrix\;
    outputMatrix = new emptyMatrix\;
    \While{$\left(\frac{remainingMatrix.Dim1}{M} \geq \epsilon \right)$}{
        \tcc{Part 1: Sort, pick, update}
        Sort(remainingMatrix)\;
        outputMatrix $\leftarrow$ SelectValidSamples(outputMatrix, remainingMatrix, targetVector, bins)\;
        \BlankLine
        \tcc{Part 2: Mix, pick, update}
        Mix(remainingMatrix)\;
        outputMatrix $\leftarrow$ SelectValidSamples(outputMatrix, remainingMatrix, targetVector, bins)\;
        
    }
    \KwRet outputMatrix\;}
\BlankLine
  \Fn{\FSort{matrix}}{
  \tcc{This method sorts each column of the matrix in ascending order}
                \For{column in matrix}    
        { 
        	column $\leftarrow$ sortVectorAscending(column)
        }
        \KwRet matrix\;
  }
    \Fn{\FMix{matrix}}{
    \tcc{This method mixes each column  of the matrix}
                \For{column in matrix}    
        { 
        	column $\leftarrow$ mixColumnRandomly(column)
        }
        \KwRet matrix\;
  }
    \Fn{\FSelect{outputMatrix, remainingMatrix,target,bins}}{
    \tcc{This method selects all samples which are deemed valid from "remainingMatrix" and stores them in "outputMatrix"}
    \tcc{First we get the index prices from the samples}
        indexPrices = sumOverRows(remainingMatrix)\;
                \For{bin in bins}    
        {       
            \tcc{Then we select all samples which fall into a bin}
                rowsInBins = remainingMatrix.where(indexPrices in bin)\;

                tcc{Now we store at most N samples of rowsInBins}
                N = targetVector(bin)\;
                selectedRows = selectMaxRows(rowsInBins,N)\;
        	outputMatrix.add(selectedRows)\;
                matrix.remove(selectedRows)\;
                target(bin) $\leftarrow$ target(bin) - count(selectedRows)
        }
        \KwRet outputMatrix\;
  }
  
\end{algorithm}
\section{Pricing baskets with assets not contained in $N_I$}
The proposed methodology focuses on pricing derivatives for a basket of assets that are contained in the index. However, it can be extended to handle derivatives on baskets of constituents not part of the index but presumed to exhibit a similar correlation structure.

Consider a stock $S_e(t)$ where $e > N$. Suppose we aim to price derivatives on a basket $N_B^+ = N_B \cup {e}$, with $N_B \subset \N_{\leq N}$ representing a specific basket subset. Since there is no available information about the correlation structure of $S_e$ with another asset in $N_B$, coherent pricing of these derivatives is not possible. Yet, if we anticipate comparable correlation characteristics between $S_e$ and the other basket constituents (for instance, if the stocks share the same industry), we can employ the calibrated copula from the index to also include the asset $e$.

To implement this, consider $s_n(t)$ as samples of a constituent (following the order obtained after a certain step, e.g., step ii)), with $n \notin N_B$ to ensure these samples are not used in the basket calibration. Applying the probability distribution function of $n$ to these samples allows us to obtain uniformly distributed samples.
\begin{equation}
    u_n(t) = F_{S_n(t)}(s_n(t)).
\end{equation}
As the ordering of $u_n(t)$ remains consistent with other samples, the copula remains unchanged. By projecting these samples using the probability distribution of $S_e$, we derive samples represented as
\begin{equation}
    s_e(t) = F_{S_e(t)}( u_n(t) ),
\end{equation}
which are now used to obtain the basket samples
\begin{equation}
    b(t) = \sum_{n \in N_B^+}s_n(t).
\end{equation}
This approach thus allows for the construction of baskets with up to $N$ external components.
\section{Proof of \Cref{lem:nonempty}}
\label{app:proof}
The proof of \Cref{lem:nonempty} is the theoretical justification of a rearrangement algorithm. Although conceptually not difficult, the proof is technically involved. We will show the proof for a simplified model of two assets since the generalization to multiple assets is trivial. Consider a model of two assets $S_1, S_2$. Since we only consider a fixed time, we drop the time indication. The probability distribution of the sum of $S_1$ and $S_2$ is given by $F_{S_1+S_2}(x)$.

For the rearrangement algorithm, we start with two independent random variables, which are equal in distribution to $S_1$ and $S_2$. We denote them as $\bar{S}_1$, $\bar{S}_2$ such that $F_{S_1} = F_{\bar{S}_1}$ and $F_{S_2} = F_{\bar{S}_2}$. The proof of the theorem is simple: We first assume that we can obtain samples of $S_1$ and $S_2$ directly. Since these samples have all the desired properties, we will show that we can find an ordering $\pi$, such that $\pi(\bar s_1),\pi(\bar s_2)$ is very close to the samples $s_1, s_2$ obtained from $S_1$ and $S_2$

The first observation we make is that if $[s]_M$ is a sample of size $M$, the empirical distributions $\hat{F}_{\bar{S}_n, [s]_M}$ and $\hat{F}_{S_n, [s]_M}$ converge uniformly almost surely, as $M$ goes to infinity. This is a consequence of the Glivenko-Cantelli theorem since both empirical distributions converge uniformly a.s. to the same probability distribution. 

This means that we can fix $\epsilon$ and with probability 1, samples of $[s]_M$ of size $M$ (of all 4 random variables) will be such that
\begin{equation}
   \left|\hat{F}_{S_n,  [s]_M}(x) - \hat{F}_{\bar{S}_n,  [s]_M}(x)\right | < \epsilon, \hspace{0.5cm} \forall x \in \R, n \in \{1,2\}.
\end{equation}

Let us now define a permutation on $\bar{s}_n$, which will be used for this purpose. We first define the ordering per"mutation $o_{s_n} \colon \N_M \to \N_M $ as the permutation which ordered a sample vector in ascending order:
\begin{equation}
    o_{s_n}(s_n) = (s_n^{(1)},s_n^{(2)},\dots, s_n^{(M)}), 
\end{equation}
where $s_n^{(1)}$ and $s_n^{(1)}$ are the smallest and largest sample of $s_n$, respectively
For both $n \in \{1,2\}$, we define $\pi_n$ as:
\begin{equation}
    \pi_n(\bar{s}_n) = o^{-1}_{s_n}  o_{\bar{s}_n}(\bar{s}_n).
\end{equation}
The proof now consists of showing that the following quantity can be bounded based on $\epsilon$ for all $x$, so that it converges to $0$ as $M \to \infty$
\begin{equation}
    E(x) := \left|\hat{F}_{S_1+S_2, [s]_M}(x) - \hat{F}_{\bar{S}_1+\bar{S}_2, \pi([s]_M)}(x) \right|.
\end{equation}
Based on the bound of the marginal distributions, we will now make the following assumption, which leads to the worst possible bound of $E(x)$:
\begin{equation}
   \hat{F}_{S_n,  [s]_M}(x)  + \epsilon =  \hat{F}_{\bar{S}_n,  [s]_M}(x), \hspace{0.5cm} \forall x \in \R, n \in \{1,2\}.
\end{equation}
Without loss of generality, we chose  $\hat{F}_{S_n,  [s]_M}(x) <  \hat{F}_{\bar{S}_n,  [s]_M}(x) $. We find
\begin{equation}
    \hat{F}_{S_1+S_2,  [s]_M}(x)  = \int_{0}^\infty \frac{1}{M}\left|\{ m \leq M : s^m_1 \leq y , s^m_2 \leq x-y \}\right|\d y
\end{equation}
and 
\begin{equation}
    \hat{F}_{\bar{S}_1+\bar{S}_2, \pi([s]_M)}(x)  = \int_{0}^\infty \frac{1}{M}\left|\{ m \leq M : \bar{s}^m_1 \leq y , \bar{s}^m_2 \leq x-y \}\right|\d y.
\end{equation}
The difference between the two quantities depends on the size of the set under the integral:
\begin{equation}
    E(x) = \int_0^\infty \frac{1}{M} \left|\{ m \leq M : s^m_1 \leq y , s^m_2 \leq x-y \}\right| -  \left|\{ m \leq M : \bar{s}^m_1 \leq y , \bar{s}^m_2 \leq x-y \}\right| \d y
\end{equation}
This term under the integral is of the following form $|A \cap B| - |C \cap D |$, where $C \subset A$ and $D \subset B$ due to the assumption above. With the general equation $|A \cap B | + |A \cup B| = |A| + |B|$, we then obtain:
\begin{align}
    |A \cap B| - |C \cap D | &= |A| + |B| - |A \cup B| - (|C| + |D| - |C \cup D| )\\
    &= |A| - |D| + |B| - |C| - (|A \cup B| -  |C \cup D| )
\end{align}
We observe now that since $(|A \cup B| -  |C \cup D| )$ is certainly positive, it is also bounded by
\begin{equation}
    (|A \cup B| -  |C \cup D| ) \leq |A| - |D| + |B| - |D|.
\end{equation}
For this reason, we can obtain the previous result that
\begin{equation}
    |A \cap B| - |C \cap D | \leq |A| - |D| + |B| - |D| < 2\epsilon.
\end{equation}
Therefore, combining all the results, we have that
\begin{equation}
    E(x) \leq \int_0^\infty \frac{1}{M} \left| |A \cap B| - |C \cap D | \right| \leq \int_0^T 2 \epsilon.
\end{equation}
Since the samples $s_1,s_2,\bar{s}_1, \bar{s}_2$ are almost surely bounded, we obtain that
\begin{equation}
    E(x) \leq 2\epsilon B
\end{equation}
for some $B < \infty$
\end{document}